\documentclass[onecolumn,superscriptaddress,floatfix,showpacs]{revtex4-2}

\usepackage{tabularx}

\usepackage[utf8]{inputenc}
\usepackage[T1]{fontenc}     
\usepackage[british]{babel}  
\usepackage[sc,osf]{mathpazo}
\usepackage[scaled=0.86]{berasans}  
\usepackage[colorlinks=true, linkcolor=blue,citecolor=blue, urlcolor=blue]{hyperref}  
\usepackage{mathtools}
\usepackage{lipsum}
\usepackage{graphicx} 
\usepackage{subfig}
\usepackage[babel]{microtype}  
\usepackage{amsmath,amssymb,amsthm,bm,amsfonts,mathrsfs,bbm}

\renewcommand{\qedsymbol}{\rule{0.7em}{0.7em}}
\usepackage{xspace}  
\usepackage{pgfplots}
\usepackage{xcolor,colortbl}
\def\ba{\begin{equation}}
	\def\ea{\end{equation}}
\def\bea{\begin{eqnarray}}
	\def\eea{\end{eqnarray}}
\def\ben{\begin{equation*}}
	\def\een{\end{equation*}}
\def\bean{\begin{eqnarray*}}
	\def\eean{\end{eqnarray*}}
\def\bma{\begin{mathletters}}
	\def\ema{\end{mathletters}}
\def\bi{\begin{itemize}}
	\def\ei{\end{itemize}}

\newcommand{\be}{\begin{equation}}
	\newcommand{\ee}{\end{equation}}

\newcommand{\ket}[1]{\ensuremath{|#1\rangle}}

\newcommand{\bra}[1]{\ensuremath{\langle#1|}}

\newcommand{\kommentar}[1]{}

\newcommand{\forget}[1]{}

\usepackage{amsthm}

\newtheorem{definition}{Definition}

\newtheorem{observation}{Observation}

\newtheorem{theorem}{Theorem}
\newtheorem{property}{Property}

\newtheorem{lemma}{Lemma}

\usepackage{pgfplots}
\pgfplotsset{compat=1.18}
\begin{document}
    \title{
    Fidelity of entanglement and quantum entropies: unveiling their relationship in quantum states and channels
    }
    \author{Komal Kumar}
    \email{p20210063@hyderabad.bits-pilani.ac.in}
    \affiliation{Department of Mathematics, Birla Institute of Technology and Science Pilani, Hyderabad Campus, Jawahar Nagar, Kapra Mandal, Medchal District, Telangana-500078, India}
    
    \author{Bivas Mallick}
    \email{bivasqic@gmail.com}
    \affiliation{S. N. Bose National Centre for Basic Sciences, Block JD, Sector III, Salt Lake, Kolkata, 700 106, India}
    
    \author{Tapaswini Patro}
    \email{tapaspatro44@gmail.com}
    \affiliation{ Department of Physics, Indian Institute of Technology Hyderabad, Kandi, Sangareddy, Telangana 502285, India}
    \affiliation{ Society for Electronic Transaction and Security (SETS), MGR Knowledge City, CIT Campus, Taramani, Chennai, Tamil Nadu-600113, India}

    \author{Nirman Ganguly}
    \email{nirmanganguly@hyderabad.bits-pilani.ac.in, nirmanganguly@gmail.com}
    \affiliation{Department of Mathematics, Birla Institute of Technology and Science Pilani, Hyderabad Campus, Jawahar Nagar, Kapra Mandal, Medchal District, Telangana-500078, India}
	
        \begin{abstract}
        
Entanglement serves as a fundamental resource for various quantum information processing tasks. Fidelity of entanglement (which measures the proximity to a maximally entangled state)  and various quantum entropies are key indicators for certifying entanglement in a quantum state. Quantum states with high fidelity are particularly useful for numerous information-theoretic applications. Similarly, states possessing negative conditional entropy provide significant advantages in several quantum information processing protocols. In this work, we examine the relationship between these two indicators of entanglement, both in state and channel regimes. First, we present a comprehensive analysis and characterization of channels that reduce fidelity of entanglement beyond a threshold limit of bipartite composite systems. In this context, we introduce the notion of fidelity annihilating channel and discuss its topological characterization, along with various information-theoretic properties. We then provide a comparison between channels that diminish the fidelity of entanglement and negative conditional entropies, using the depolarizing channel as an illustrative example. In particular, we determine the parameter regimes in which the depolarizing channel belongs to a given family and establish connections among these families of channels.  Extending our analysis from channels to the state level, we further examine the relationship between the fidelity of entanglement and various quantum entropies for general two-qubit states.  We derive the upper bound on Rényi $2$-entropy, conditional Rényi $2$-entropy, Tsallis $2$-entropy, and conditional Tsallis $2$-entropy, in terms of the fidelity of entanglement. Finally, we explore the relationship between relative entropy and the fidelity of entanglement of a two qudit quantum state.
        \end{abstract}
	\maketitle

    \section{Introduction}\label{intro}
        Quantum entanglement constitutes a key concept in quantum theory, underpinning the development and advancement of the rapidly evolving domain of quantum technologies \cite{peres1997quantum,bohr1935can,einstein1935can,horodecki2009quantum}. Extensive research has been dedicated to the study of quantum entanglement, which offers significant advantages in various quantum information processing tasks \cite{ekert1991quantum,branciard2012one,pati2000minimum,hillery1999quantum,cleve1999share,bandyopadhyay2000teleportation,bennett1993teleporting}. Despite several diverse and interesting applications of entanglement, it is essential to certify entanglement in a quantum state \cite{lewenstein2001characterization,terhal2000bell,guhne2009entanglement,ganguly2009witness,ganguly2014witness,ganguly2013common,mallick2024genuine,mallick2025higher,mallick2024genuine,mukherjee2025efficient}, since, not all entangled states can enhance the effectiveness of information processing tasks. Quantum conditional entropies \cite{cerf1999quantum} and fidelity of entanglement \cite{horodecki1999general} (also known as singlet fraction or fully entangled fraction) serve as crucial indicators of entanglement, offering deeper insights into the potential of quantum states as resources for various quantum technologies.
     
         Entropy captures the general notion of the amount of uncertainty or randomness of random variables corresponding to their probability distribution. In other words, it is the expectation of the information content of the random variable. This concept is formalized through Shannon entropy, which serves as a standard measure and can be extended to other forms such as Rényi $\alpha$-entropy \cite{renyi1961measures, linden2013structure} and Tsallis $\alpha$-entropy \cite{tsallis1988possible}.  Moving from the classical to the quantum domain, this notion is captured by quantum entropy \cite{muller2013quantum, von2018mathematical}, which reflects both classical and quantum uncertainty in a quantum state.
         Several quantum entropy inequalities, including entropic Bell inequalities \cite{cerf1997entropic}, entropic steering inequalities \cite{schneeloch2013einstein, costa2018entropic}, and the entropic uncertainty relation \cite{wehner2010entropic}, are established using quantum entropies. Conditional entropies are always non-negative in classical systems. However, in the quantum context, they can become negative \cite{cerf1997negative}, representing a substantial deviation from classical principles. While states with negative conditional entropy are entangled, the converse is not true i.e., there are entangled states which have a non-negative conditional entropy. Several studies have demonstrated the operational interpretation and utility of conditional quantum entropies as a resource through quantum protocols \cite{vempati2021witnessing, vempati2022unital, brandsen2021quantum, patro2017non,bennett1992communication, bruss2004distributed, prabhu2013exclusion,horodecki2005partial, horodecki2007quantum}. Another important quantifier for certifying the correlation between quantum systems is the fidelity of entanglement. Fidelity of entanglement measures the proximity to a maximally entangled state. Quantum states in $d \otimes d$ systems, whose fidelity of entanglement is greater than $\frac{1}{d}$ act as resource, and have operational implications in several quantum information processing tasks such as quantum key distribution \cite{curty2004entanglement}, quantum teleportation \cite{bennett1993teleporting}, entanglement swapping \cite{zukowski1993event}, quantum cryptography \cite{gisin2002quantum}, remote state preparation \cite{bennett2001remote, shi2002remote, ye2004faithful} etc. Recently, substantial advancements have been made in investigating the properties of the fidelity of entanglement within various aspects of quantum technologies \cite{rui2010estimation,li2008upper,huang2016upper,zhao2015maximally,Zhao_2010,grondalski2002fully,patro2022absolute,ganguly2011entanglement,cavalcanti2013all,ghosal2025repeater}.

 Quantum state that is essential for the success of a specific quantum protocol is termed as a resource, whereas a quantum state that does not contribute to the protocol is referred to as a free state. Quantum resources undergo physical evolution and interact with their environment, which diminishes their effectiveness as resources, transforming them into free states. This evolution and interaction with the environment are mathematically described by a completely positive, trace-preserving map, known as a quantum channel. In \cite{pal2015non, luo2022coherence, heinosaari2015incompatibility, ku2022quantifying, horodecki2003entanglement, moravvcikova2010entanglement, chruscinski2006partially, patro2024quantum, mallick2024characterization}, comprehensive research have been performed on different quantum channels, with a focus on resource-theoretic characteristics. In the resource theory of negative conditional entropy, states characterized by negative conditional entropy are termed as resources, while states with non-negative conditional entropy are known as free states. Quantum channels that destroy the negativity of conditional entropy across the bipartition are known as negative conditional entropy breaking channels \cite{srinidhi2024quantum,muhuri2023information}. Whereas negative conditional entropy annihilating channels are those that destroy the negative conditional entropy within a subsystem \cite{srinidhi2024quantum}. On a different note, the quantum states with the fidelity of entanglement greater than some threshold are useful for teleportation, i.e., act as resources. Quantum channels that diminish the fidelity of entanglement beyond the threshold are known as fidelity breaking channels $(\mathbb{FBC})$, where fidelity means entanglement fidelity. It is worth mentioning that determining whether a quantum channel is entanglement-breaking, requires only evaluating its action on a maximally entangled state \cite{horodecki2003entanglement}. This significantly simplifies the characterization of such channels. However, this property does not hold true for fidelity-breaking or negative conditional entropy-breaking channels, hence a more comprehensive analysis over a broader set of input states is required for these channels. In this regard, several interesting properties of these channels have been previously studied in \cite{srinidhi2024quantum,muhuri2023information}. In this work, we extend that study by providing deeper insights into some additional properties of fidelity breaking channels. Furthermore, we introduce the concept of fidelity annihilating channels, providing a topological characterization and an in-depth investigation of their features. Furthermore, we present examples of these channels by analyzing the depolarizing channel for both qubit and qutrit systems. We then investigate the relationship between fidelity breaking and annihilating channels with negative conditional entropy breaking and annihilating channels.
 
 \par Motivated by the above connection between channels diminishing the fidelity of entanglement and negative quantum conditional entropies (QCE), as well as the fact that fidelity of entanglement itself is a crucial indicator of entanglement in a quantum state, we aim to explore the relationship between these two fundamental quantifiers of quantum correlations at the state level. In recent studies, the relation between quantum nonlocality and negative quantum conditional entropy (QCE) has been explored \cite{friis2017geometry, kumar2023quantum}. Specifically, in \cite{friis2017geometry}, the authors examined the relationship between the Bell-CHSH inequality \cite{clauser1969proposed} and QCE. Later in \cite{kumar2023quantum}, the relation between the CJWR steering inequality \cite{cavalcanti2009experimental} and QCE has been explored. More recently, in \cite{kumar2024quantum}, authors have explored this relationship for specific classes of states, such as Weyl states and Werner states. In this work, we extend this analysis by examining the analytical connection between quantum entropies and the fidelity of entanglement for a general two-qubit state. We show that R\'enyi $2$-entropy of the general two-qubit state is bounded above provided that its fidelity of entanglement is greater than $\frac{1}{2}$, and vice versa. Also, for the general two-qubit state, we obtain an upper bound on the conditional R\'enyi $2$-entropy, which serves as the necessary and sufficient condition for its fidelity of entanglement to be greater than $\ \frac {1}{2}$. We also observe that the min-R\'enyi entropy and conditional min-R\'enyi entropy of the general two-qubit state are both upper-bounded whenever its fidelity of entanglement exceeds $\frac{1}{2}$. Additionally, we obtain the upper bounds for the Tsallis $2$-entropy that is a necessary and sufficient condition for its fidelity of entanglement to be greater than $\frac{1}{2}$, and the same holds for the conditional Tsallis $2$-entropy of general two-qubit states. Finally, we extend our study to explore the interplay between relative entropy and the fidelity of entanglement in $d \otimes d$ systems.
 \par The article is structured as follows: In Section \ref{II}, we present a concise overview of the key preliminaries related to the Bloch-Fano decomposition of density matrices, quantum entropies, and fidelity of entanglement, along with discussions on fidelity breaking channels and negative conditional entropy-breaking channels. In Section \ref{III}, we analyze the properties of fidelity breaking channels and introduce the concept of fidelity annihilating channels. We further examine their topological characterization and explore the relationships between fidelity breaking channels and annihilating channels, as well as their connection to negative conditional entropy-breaking and annihilating channels. We then examine the relation between quantum entropies, quantum conditional entropies, and fidelity of entanglement for the general two-qubit state in section \ref{IV}. Section \ref{V} then explores the relationship between relative entropy and fidelity of entanglement in $  d\otimes d$ systems. Finally, section \ref{VI} provides concluding remarks and discusses potential directions for future research.

		\section{Preliminaries}\label{II}

        \begin{table}[h]
    \centering
    \renewcommand{\arraystretch}{1.3} 
    \setlength{\tabcolsep}{15pt}      
    
    \begin{tabular}{|c|c|}
        \hline
        \textbf{Description} & \textbf{Notations}  \\
        \hline
        Set of quantum states & $\mathcal{D}(\mathcal{H}_{AB})$  \\ 
          \hline
        Identity matrix & $\mathbb{I}$  \\ 
        \hline
        Identity map & $\mathcal{I}$  \\ 
        \hline
        Quantum channel & $\mathcal{N}$  \\ 
        \hline
        Trace norm & $||\cdot||_{1}$  \\ 
        \hline
        Frobenius norm & $||\cdot||_{2}$  \\ 
        \hline
        Operator norm & $||\cdot||_{O}$  \\ 
        \hline
        Set of fidelity annihilating channels & $\mathbb{FAC}$  \\ 
        \hline
        Set of fidelity breaking channels & $\mathbb{FBC}$  \\
        \hline
        Set of negative conditional entropy breaking channels & $\mathbb{NCEBC}$  \\
        \hline
           Set of negative conditional entropy annihilating channels & $\mathbb{NCEAC}$ \\
        \hline
        Set of entanglement breaking channels & $\mathbb{EB}$\\
        \hline
        Set of entanglement annihilating channels & $\mathbb{EA}$\\
        \hline
    \end{tabular}
    
    \captionsetup{position=bottom} 
    \caption{The table presents the various notations used throughout this work}
    \label{tab:example}
\end{table}

        This section presents the fundamental tools that serve as the foundation of our study. Throughout the paper, we use the standard notations and terminology widely used in quantum information theory.

		\subsection{Bloch-Fano Decomposition of density matrices}
		
		In bipartite composite quantum systems, the quantum state is described by density matrices, which can be represented as \cite{friis2017geometry}
			\begin{eqnarray}
				&&\rho_{AB} =\frac{1}{d_A d_B} [\mathbb{I}_A \otimes \mathbb{I}_B + \sum_{i=1}^{d_A^2 - 1} a_i g_i^A \otimes \mathbb{I}_B + \sum_{j=1}^{d_B^2 - 1} b_j  \mathbb{I}_A \otimes g_j^B + \sum_{i=1}^{d_A^2 - 1} \sum_{j=1}^{d_B^2 - 1} t_{ij} g_i^A \otimes g_j^B], \label{e1}
			\end{eqnarray}
		where $ \text{dim}~ \mathcal{H}_{A} = d_A  $ and $ \text{dim}~ \mathcal{H}_{B} = d_B  $. The hermitian operators $ g_i^k $ for $ k=A,B $ are generalizations of the Pauli matrices, i.e., they are orthogonal $ \text{Tr}[g_i^k g_j^k] = 2 \delta_{ij} $ and traceless, $ \text{Tr}[g_i^k] =0  $ and for single qubit systems they coincide with the Pauli matrices. The coefficients $ a_i, b_j \in \mathbb{R} $ are the components of the generalized Bloch vectors $ \vec{a}, \vec{b} $ of the subsystems $ A, B$, respectively. The real coefficients $ t_{ij} $ are the components of the correlation matrix $T$ such that correlation matrix is given as
        \begin{center}
        \begin{equation} \label{cmat}
             T =
            \begin{bmatrix}
              t_{11} & t_{12} & \cdots & t_{1,d_B^{2}-1} \\[6pt]
              t_{21} & t_{22} & \cdots & t_{2,d_B^{2}-1} \\[6pt]
              \vdots & \vdots & \ddots & \vdots \\[6pt]
              t_{d_A^{2}-1,1} & t_{d_A^{2}-1,2} & \cdots & t_{d_A^{2}-1,d_B^{2}-1}
            \end{bmatrix}
        \end{equation}
        
        \end{center}
        For two-qubit systems the density matrices can be represented as, 
		\begin{eqnarray}{\label{2qgs}}
			\rho_{AB}= \frac{1}{4} [\mathbb{I}_{2}\otimes \mathbb{I}_{2}+\vec{\mathfrak{a}}.\vec{\sigma}\otimes \mathbb{I}_{2} + \mathbb{I}_{2} \otimes \vec{\mathfrak{b}}.\vec{\sigma} + \sum_{i=1}^{3} \sum_{j=1}^{3} t_{ij} \sigma_i^A \otimes \sigma_j^B ].\label{e2}
		\end{eqnarray} 
		
		An interesting class of states is the \textit{locally maximally mixed states or Weyl states} which in the two-qudit systems (up to local unitaries) is given by, 
		
		\begin{equation}
			\rho_{AB} = \frac{1}{d^2} \left[\mathbb{I}_{A} \otimes \mathbb{I}_{B} + \sum_{i=1}^{d^2 - 1} w_{i} g_i^A \otimes g_i^B\right].\label{e3}
		\end{equation}
		The reduced marginals of such states are maximally mixed, i.e., $ \text{Tr}_A [\rho_{AB}] =\dfrac{\mathbb{I}_{A}}{d}  $ and $  \text{Tr}_B [\rho_{AB}] =\dfrac{\mathbb{I}_{B}}{d}.$ The two-qubit Weyl state is given by 
            \begin{equation}
			\rho_{AB}= \frac{1}{4}[\mathbb{I}_{2} \otimes \mathbb{I}_{2}  +\sum_{i=1}^{3} t_{i}(\sigma_{i} \otimes \sigma_{i})].
		\end{equation}
		\subsection{Quantum entropies}
	\par Information content in a state is quantified through entropies. In this subsection, we discuss a list of quantum entropies that will be used in our work. The most generalized version of quantum entropy is known as the Rényi entropy. The Rényi $\alpha$-entropy for a given density matrix $\rho_{AB}$ corresponding to system $AB$ is given by
		\begin{align}
			S_{\alpha}(AB)= \frac{1}{1-\alpha} \log_{2}\left[\text{Tr}(\rho_{AB}^{\alpha})\right], \alpha >0, \alpha \neq 1. \label{e5}
		\end{align}
		 The corresponding conditional R\'enyi $\alpha$-entropy is given by $ S_\alpha(A|B)=S_\alpha(AB)-S_\alpha(B) $. There are various types of R\'enyi entropy depending on the quantity $\alpha$. Following are entropic quantities that we have used in our work:
	\begin{itemize}
            \item \textit{\textbf{von Neumann entropy:}} The von Neumann entropy is the limiting case of the Re\'nyi $\alpha$-entropy as 
            $\alpha\rightarrow 1$. The von Neumann entropy of a quantum state $\rho_{AB} \in  \mathcal{D}(\mathcal{H}_{AB}) $  is defined as,
            \begin{equation}
                S(AB) = -\text{Tr} \left[\rho_{AB} \log_2 \rho_{AB}\right], \label{e4}
            \end{equation}
            where the logarithms are taken to the base $2$.
            The von Neumann entropy has a special association with the eigenvalues of the density matrix, i.e., it is a function of the eigenvalues. The corresponding conditional von Neumann entropy is given by $ S(A|B)=S(AB)-S(B) $.
            \item \textit{\textbf{R\'enyi $2$-entropy:}} If $\alpha=2$, R\'enyi $\alpha$-entropy is known as R\'enyi $2$-entropy, denoted by $S_{2}(AB)$. The corresponding conditional R\'enyi $2$-entropy is given by $ S_{2}(A|B)=S_{2}(AB)-S_{2}(B) $.
            \item \textit{\textbf{Min R\'enyi entropy:}} If $\alpha\rightarrow \infty$, R\'enyi $\alpha$-entropy is known as Min R\'enyi entropy, denoted by $S_{\infty}(AB)$, and is given as $S_{\infty}(AB)=-\log_{2}||\rho_{AB}||_{O}$ where $||\rho_{AB}||_{O}$ is operator norm of state $\rho_{AB}$ i.e., largest eigenvalue of state $\rho_{AB}$ \cite{linden2013structure}.
	\end{itemize}
    \par Tsallis $\alpha$-entropy of state $\rho_{AB}$ is given by 
		\begin{align}
			S^{\mathbb{T}}_{\alpha}(AB)= \frac{1}{1-\alpha} [\text{Tr}(\rho_{AB}^{\alpha})-1], \alpha >0, \alpha \neq 1. \label{e6}
		\end{align}
		The conditional Tsallis $\alpha$-entropy is given by \cite{vollbrecht2002conditional}, 
		\begin{align*}
			S^{\mathbb{T}}_{\alpha}(A|B)=\frac{\text{Tr}(\rho_B^\alpha)-\text{Tr}(\rho_{AB}^\alpha)}{(\alpha-1)\text{Tr}(\rho_B^\alpha)}.
		\end{align*}
	\subsection{Fidelity of entanglement and teleportation witness}
	 In a bipartite quantum system, for any quantum state $\rho_{AB} \in  \mathcal{D}(\mathcal{H}_{AB}) $, where $\mathcal{H}_{AB}=\mathcal{H}_{A}\otimes \mathcal{H}_{B}$, $\text{dim}~ \mathcal{H}_{A} = \text{dim}~ \mathcal{H}_{B} = d$, the fidelity of entanglement $F(\rho_{AB})$ is defined as the maximal overlap of a quantum state $\rho_{AB}$ to a maximally entangled pure state. In other words, we can also say that it is a measure of proximity of a quantum state to a maximally entangled pure state. The fidelity of entanglement of state $\rho_{AB}$ is given as
  \begin{align}
      F(\rho_{AB}) = \text{max}_{|\phi\rangle\in S}\langle\phi|\rho_{AB}|\phi\rangle, \label{e7a}
  \end{align}
  where $S$ is the set of all maximally entangled pure states. Alternatively, Eq.(\ref{e7a}) can be expressed as
   \cite{horodecki1999general} 
	 \begin{align}
   	 	F(\rho_{AB}) = \max_{\substack{U}} \langle\psi|(U\otimes \mathbb{I})\rho_{AB} (U^{\dagger}\otimes \mathbb{I}) |\psi\rangle, \label{e7}
     \end{align}
    where $|\psi\rangle=\frac{1}{\sqrt{d}}\sum_{i=0}^{d-1}|ii\rangle$ maximally entangled pure state, $ U$ is a unitary operator and $\mathbb{I}$ is the identity. In $2\otimes2$ system, fidelity of entanglement of any quantum state $\rho_{AB}$ is given as \cite{horodecki1997inseparable} 
    \begin{align}
    	F(\rho_{AB}) = \frac{1}{4}\left[1+ ||T||_{1}\right], \label{e8}
    \end{align}
     where $||T||_{1}=\text{Tr}|T|$; $|T|=\sqrt{T^{\dagger}T}$ and $T=[t_{ij}]$ is the correlation tensor of the general two-qubit density matrix $\rho_{AB}$. From \cite{horodecki1997inseparable}, we know that
     \begin{align}\label{fgh}
         F(\rho_{AB})>\frac{1}{2}  \Longleftrightarrow ||T||_{1}>1.
     \end{align}

     For any bipartite state $\rho_{AB}\in \mathcal{D}(\mathcal{H}_{AB})$, if $F(\rho_{AB})>1/d$, then the state is useful for quantum teleporation \cite{horodecki1999general,vidal2000approximate,Zhao_2010}. The set of states that are useful for teleportation is a convex and compact subset of $\mathcal{D}(\mathcal{H}_{AB})$ \cite{ganguly2011entanglement}. Consequently, by the Hahn–Banach separation theorem \cite{holmes2012geometric,rudin1976principles}, there exists a Hermitian operator $\mathcal{W}$, termed a teleportation witness, such that
         \begin{align}\label{witness}
            & \mathrm{Tr}(\mathcal{W}\rho_{AB}) < 0\, \text{for at least one state} \,\rho_{AB}\, \text{that is useful for teleportation}, \nonumber \\
            \text{and} \quad & \mathrm{Tr}(\mathcal{W}\rho_{AB}) \geq 0 \,\text{for every state} \,\rho_{AB}\, \text{that is not useful for teleportation.}
         \end{align}
In particular, for a $d \otimes d$ system, an example of such teleportation witness is given by \cite{ganguly2011entanglement}
\begin{equation} \label{tw}
    \mathcal{W} = \frac{\mathbb{I}_{d^2}}{d} - |\phi ^{+}\rangle\langle \phi^{+}|,
\end{equation}
    where $|\phi ^{+}\rangle = \tfrac{1}{\sqrt{d}}\sum_{i=0}^{d-1}|ii\rangle$.

     \subsection{Fidelity and negative conditional entropy breaking channel}
\textit{\textbf{Fidelity breaking channel:}}  In quantum information theory, certain quantum channels operate on one subsystem of a bipartite composite system $AB$ in a way that reduces its potential to serve as a resource. Specifically, some channels reduce the fidelity of entanglement of a bipartite quantum state when they are applied to one of its subsystems. These channels are referred to as fidelity breaking channels \cite{muhuri2023information}. Mathematically, a quantum channel denoted by $\mathcal{N}_{FB}$, is classified as a fidelity breaking channel if, for any bipartite quantum state $\rho_{AB} \in \mathcal{D}(\mathcal{H}_{AB})$, the resulting state $\sigma_{AB}$, obtained by applying the channel $\mathcal{N}_{FB}$ to one subsystem has a fidelity of entanglement that does not exceed $1/d$, i.e.,
        \begin{equation}
            \mathcal{I} \otimes \mathcal{N}_{FB}(\rho_{AB}) = \sigma_{AB}, \quad \text{where} \quad F(\sigma_{AB}) \leq \frac{1}{d}~,
        \end{equation}
    where $F(\sigma_{AB})$ indicates the fidelity of entanglement of the state $\sigma_{AB}$. Throughout the manuscript, we use $\mathcal{N}_{FB}$ to describe the fidelity breaking channel. \\
    
  \textit{\textbf{Negative conditional entropy breaking channel:}} 
There exists a specific kind of quantum channel that operates on one subsystem of a bipartite system, ensuring that the conditional entropy between the two subsystems remains non-negative, these channels are known as negative conditional entropy breaking channels \cite{muhuri2023information,srinidhi2024quantum}. This indicates that the channel effectively destroys any quantum correlations that could result in a negative value for conditional entropy. Formally, a quantum channel denoted by $\mathcal{N}_{NB}$, is termed a negative conditional entropy breaking channel if, for any bipartite quantum state $\rho_{AB} \in \mathcal{D}(\mathcal{H}_{AB})$, the conditional entropy $S(A|B)$ of the resulting output state $\sigma_{AB}$ is non-negative i.e.,
        \begin{equation}
            S(A|B) \geq 0,~~~~ \hspace{0.05cm} {\text{where output state}}
            \hspace{0.05cm} \sigma_{AB} = \mathcal{I} \otimes \mathcal{N}_{NB}(\rho_{AB}).
        \end{equation}
      Since negative conditional entropy is a signature of strong quantum correlations such as quantum entanglement and quantum steering, a negative conditional entropy breaking channel fundamentally limits these correlations. Throughout the manuscript, we use $\mathcal{N}_{NB}$ to describe the negative conditional entropy breaking channel.

     \section{Fidelity of entanglement and channels affecting it}\label{III}

Consider two quantum systems $A$ and $B$, associated with finite dimensional Hilbert spaces $\mathcal{H}_{A}$ and $\mathcal{H}_{B}$, respectively, both having dimension $d$. Now, consider a bipartite system $AB$ described by Hilbert space $\mathcal{H}_{AB}$. Let $\mathcal{D}(\mathcal{H}_{AB})$ is the set of quantum state associated with system $AB$, given by 
\begin{eqnarray}
   \mathcal{D}(\mathcal{H}_{AB}) = \{~\rho_{AB}~~;~~~\rho_{AB}~~\geq ~~ 0,~~\text{Tr}(\rho_{AB})~~=~~1 \}.
\end{eqnarray}
We define the sets of states whose fidelity of entanglement is less than or equal to $ 1/d$ and $  1/d`$ on systems $AB$ and $B$, by $\mathcal{D}_{\mathbb{FBC}}(\mathcal{H}_{AB})$ and $\mathcal{D}_{\mathbb{FAC}}(\mathcal{H}_{B})$, 
respectively, as follows:
\begin{eqnarray}
    \mathcal{D}_{\mathbb{FBC}}(\mathcal{H}_{AB}) = \{~\rho_{AB} \in \mathcal{D}(\mathcal{H}_{AB})~~:~~~F(\rho_{AB})~~\leq ~~ 1/d \},
\end{eqnarray}
\begin{eqnarray}
    \mathcal{D}_{\mathbb{FAC}}(\mathcal{H}_{B}) = \{~\rho_{B_1 B_2} \in \mathcal{D}(\mathcal{H}_{B})~~:~~~F(\rho_{B_1 B_2})~~\leq ~~ 1/d` \},
\end{eqnarray}
where, $\mathcal{D}_{\mathbb{FBC}}(\mathcal{H}_{AB})$ represents the set of states in $\mathcal{H}_{AB}$ whose fidelity of entanglement less than or equal to $ 1/d$, while $\mathcal{D}_{\mathbb{FAC}}(\mathcal{H}_{B})$ denotes the corresponding set of states in $\mathcal{H}_{B}$ whose fidelity of entanglement less than or equal to $ 1/d`$. In the case of fidelity annihilating channels we consider that the subsystem $B$ can be partitioned into $B_1 | B_2$, where both $B_1,B_2$ has dimension $d`$ ($ d` < d $ ). 
\subsection{Fidelity breaking channels}
\renewcommand{\proofname}{\textbf{Proof}} 
In \cite{muhuri2023information}, the fidelity breaking channel has been discussed, and it is the class of quantum channels that act on one subsystem of a bipartite system $AB$ and necessarily reduce its fidelity of entanglement. We denote the set of fidelity breaking channels by $\mathbb{FBC}$. Specifically, given any initial state of $\mathcal{D}(\mathcal{H}_{AB})$, the action of any channel in $\mathbb{FBC}$ ensures that the fidelity of entanglement of the resulting state is less than or equal to $1/d$. The set of fidelity breaking channels is expressed as 
\begin{eqnarray}
    \mathbb{FBC} = \{~\mathcal{N}_{FB}~|~F(\sigma_{AB}) \leq 1/d,~\sigma_{AB} = (\mathcal{I}\otimes \mathcal{N}_{FB})(\rho_{AB})~~ \forall~\rho_{AB}\in \mathcal{D}(\mathcal{H}_{AB}) \}.
\end{eqnarray}
In another way, we can say that a channel is said to be in $\mathbb{FBC}$ if 
\begin{eqnarray}
    (\mathcal{I} \otimes \mathcal{N}_{FB})[\mathcal{D}(\mathcal{H}_{AB})] \subset \mathcal{D}_{\mathbb{FBC}}(\mathcal{H}_{AB}).
\end{eqnarray}
\subsubsection{\textit{\textbf{{Properties of fidelity breaking channels}}}}
Now, we will discuss some fundamental properties of $\mathbb{FBC}$.

\begin{property}{\label{tb1}}
    If $\mathcal{N}_{FB}^{1}$, $\mathcal{N}_{FB}^{2} \in \mathbb{FBC}$, then $\mathcal{N}_{FB}^{1}\circ\mathcal{N}_{FB}^{2}\in \mathbb{FBC}$.
\end{property}
\begin{proof}
    Since, $\mathcal{N}_{FB}^{1}, \mathcal{N}_{FB}^{2} \in \mathbb{FBC}$, then we have 
    \begin{equation}
        F(\sigma_{AB}^{1}) \le 1/d ~~\text{and }~~ F(\sigma_{AB}^{2}) \le 1/d ~~ \text{where}~~ \sigma_{AB}^{1} =(\mathcal{I}\otimes \mathcal{N}_{FB}^{1})(\rho_{AB}),~~ \sigma_{AB}^{2} = (\mathcal{I}\otimes \mathcal{N}_{FB}^{2})(\rho_{AB})
    \end{equation}
    which further implies
    \begin{equation}\label{sc1}
        (\mathcal{I}\otimes \mathcal{N}_{FB}^{1})[\mathcal{D}(\mathcal{H}_{AB})] \subset \mathcal{D}_{\mathbb{FBC}}(\mathcal{H}_{AB}) ~~\text{and}~~(\mathcal{I}\otimes \mathcal{N}_{FB}^{2})[\mathcal{D}(\mathcal{H}_{AB})] \subset \mathcal{D}_{\mathbb{FBC}}(\mathcal{H}_{AB})
    \end{equation}
    Now, consider the serial concatenation, $\mathcal{N}_{FB}^{1}\circ\mathcal{N}_{FB}^{2}$ :
    \begin{equation}\label{sc2}
        [\mathcal{I}\otimes (\mathcal{N}_{FB}^{1}\circ\mathcal{N}_{FB}^{2})](\rho_{AB})~=~ (\mathcal{I}\otimes \mathcal{N}_{FB}^{1})(\mathcal{I}\otimes \mathcal{N}_{FB}^{2})(\rho_{AB})
    \end{equation}
    Using Eq.(\ref{sc1}) and Eq.(\ref{sc2}) we obtain
    \begin{equation}
        [\mathcal{I}\otimes (\mathcal{N}_{FB}^{1}\circ\mathcal{N}_{FB}^{2})](\mathcal{D}(\mathcal{H}_{AB}))~=~ (\mathcal{I}\otimes \mathcal{N}_{FB}^{1})(\mathcal{I}\otimes \mathcal{N}_{FB}^{2})[\mathcal{D}(\mathcal{H}_{AB})]\subset \mathcal{D}_{\mathbb{FBC}}(\mathcal{H}_{AB})
    \end{equation}
    Hence, $\mathcal{N}_{FB}^{1}\circ\mathcal{N}_{FB}^{2} \in \mathbb{FBC}.$
    \renewcommand{\qedsymbol}{}
    
\end{proof}

\begin{property}\label{tb2}
    If $\mathcal{N}_{FB} \in \mathbb{FBC}$ and $\mathcal{T}$ be any arbitrary channel then $\mathcal{N}_{FB}\circ\mathcal{T} \in \mathbb{FBC} $.
\end{property}
\begin{proof}
    Since, $\mathcal{N}_{FB}\in \mathbb{FBC}$, then we have
    \begin{equation}
        F(\sigma_{AB}) \leq 1/d~~ \text{where}~~ \sigma_{AB} = \mathcal{N}_{FB}(\rho_{AB}),~~\text{for all} ~~\rho_{AB} \in \mathcal{D}(\mathcal{H}_{AB})
    \end{equation}
    which further implies
    \begin{equation}\label{scb1}
        (\mathcal{I} \otimes \mathcal{N}_{FB})[\mathcal{D}(\mathcal{H}_{AB}) ] \subset \mathcal{D}_{\mathbb{FBC}}(\mathcal{H}_{AB})
    \end{equation}
    Now, consider  $\mathcal{N}_{FB}\circ\mathcal{T}:$
    \begin{equation}\label{scb2}
        [\mathcal{I}\otimes (\mathcal{N}_{FB}\circ\mathcal{T})](\rho_{AB})~=~ (\mathcal{I}\otimes \mathcal{N}_{FB})(\mathcal{I}\otimes \mathcal{T})(\rho_{AB})= (\mathcal{I}\otimes \mathcal{N}_{FB})(\sigma\;'_{AB})~~ \text{where}~~ (\mathcal{I}\otimes \mathcal{T})(\rho_{AB})=\sigma\;'_{AB}
    \end{equation}
    Using Eq.(\ref{scb1}) and Eq.(\ref{scb2}), we obtain
    \begin{equation}
        [\mathcal{I}\otimes (\mathcal{N}_{FB}\circ\mathcal{T})](\mathcal{D}(\mathcal{H}_{AB}))~=~ (\mathcal{I}\otimes \mathcal{N}_{FB})[\mathcal{D}(\mathcal{H}_{AB})]\subset \mathcal{D}_{\mathbb{FBC}}(\mathcal{H}_{AB})
    \end{equation}
    Hence, $\mathcal{N}_{FB}\circ\mathcal{T} \in \mathbb{FBC}.$\\
    \renewcommand{\qedsymbol}{}
\end{proof}

\subsection{Fidelity annihilating channels}
In this subsection, we introduce fidelity annihilating channels and then explore their characteristics and fundamental properties. We denote the set of fidelity annihilating channels by $\mathbb{FAC}$. Fidelity breaking channels are those that diminish the fidelity of entanglement across the partition $A | B$, whereas fidelity annihilating channels reduce the fidelity of entanglement within the bipartite subsystem $B$. That is, we consider that $B$ can be partitioned into $B_1 | B_2$, where dimension of both $B_1, B_2$ is $ d` $. The set of fidelity annihilating channels can be categorized into two types: local and non-local. We begin by defining non-local and local fidelity annihilating channels.

\begin{definition}
    Consider that subsystem $B$ can be partitioned into $B_1 | B_2$, where dimension of both $B_1, B_2$ is $ d` $
    Then, a channel denoted by $\mathcal{N}_{FA}$ acting on bipartite subsystem $B$ is said to be a fidelity annihilating channel if $F(\sigma_{B})\leq 1/d`$ where $\sigma_{B}=\mathcal{N}_{FA}(\rho_{B})$, for all $\rho_{B}\in \mathcal{D}(\mathcal{H}_{B})$.
\end{definition}

\begin{definition}
    A fidelity annihilating channel $\mathcal{N}_{FA}$ is $2$-local if it can be expressed as $\mathcal{N}_{FA}= \mathcal{N}^{1} \otimes \mathcal{N}^{2}$. Otherwise, it is termed as a non-local fidelity annihilating channel and the corresponding set is denoted by non-local$-\mathbb{FAC} $.    
\end{definition}

\subsubsection{\textit{\textbf{Characterization of non-local fidelity annihilating channels}}} Here, we investigate the topological characterization of non-local fidelity annihilating channels $(\mathbb{FAC})$.
\begin{theorem}
    $\text{The set non-local}-\mathbb{FAC}$ is convex.
\end{theorem}
\begin{proof}
    Since,~~ $\mathcal{N}_{FA}^{1},~\mathcal{N}_{FA}^{2} \in \text{non-local}-\mathbb{FAC}$,~then for any $\rho_{B}\in\mathcal{D}(\mathcal{H}_{B})$ we have
    \begin{equation}\label{ce1}
        F(\sigma_{B}^{1}) \leq 1/d` ~~\text{and }~~ F(\sigma_{B}^{2}) \leq 1/d` ~~ \text{where}~~ \sigma_{B}^{1} = \mathcal{N}_{FA}^{1}(\rho_{B}),~~ \sigma_{B}^{2} = \mathcal{N}_{FA}^{2}(\rho_{B})
    \end{equation}
    We need to show that
    $~~\mathcal{N}_{FA}= \lambda\mathcal{N}_{FA}^{1}+(1-\lambda)\mathcal{N}_{FA}^{2} \in \text{non-local}-\mathbb{FAC}.$\\
   Applying $\mathcal{N}_{FA}$ on any arbitrary density operator $\rho_{B}\in\mathcal{D}(\mathcal{H}_{B})$, we have
   \begin{align}\label{ce2}
       \sigma_{B} = \mathcal{N}_{FA}(\rho_{B})= (\lambda\mathcal{N}_{FA}^{1}+(1-\lambda)\mathcal{N}_{FA}^{2})(\rho_{B})= \lambda\mathcal{N}_{FA}^{1}(\rho_{B})+(1-\lambda)\mathcal{N}_{FA}^{2}(\rho_{B})=\lambda\sigma_{B}^{1}+(1-\lambda)\sigma_{B}^{2}
   \end{align}
   Using Eq.(\ref{ce1}), Eq.(\ref{ce2}), and result from \cite{ganguly2011entanglement}, we obtain
   \begin{equation}
       F(\sigma_{B}) = F(\lambda\sigma_{B}^{1}+(1-\lambda)\sigma_{B}^{2}) = \lambda F(\sigma_{B}^{1})+(1-\lambda)F(\sigma_{B}^{2}) \leq 1/d`
   \end{equation}
   Therefore, $\mathcal{N}_{FA} \in \text{non-local}-\mathbb{FAC}.$ Hence, $\text{the set non-local}-\mathbb{FAC}$~~is convex. \renewcommand{\qedsymbol}{}
    
\end{proof}

\begin{theorem}
    $\text{The set non-local}-\mathbb{FAC}$ is compact.
\end{theorem}

\begin{proof} Firstly, to prove that the set $\text{non-local}-\mathbb{FAC}$ is closed, it suffices to show that it contains all of its limit points. Since, the set $\text{non-local}-\mathbb{FAC}$ is convex and contains at least two elements, it necessarily possesses at least one limit point. Let $\mathcal{N}_{FA}^{0}$ be an arbitrary limit point of $\text{non-local}-\mathbb{FAC}$. Consider an open ball $B_{\frac{1}{r}}(\mathcal{N}_{FA}^{0})$ of radius $\frac{1}{r}$ centered on $\mathcal{N}_{FA}^{0}$, then by the definition of a limit point, we have
\begin{equation}
   \{ B_{\frac{1}{r}}(\mathcal{N}_{FA}^{0}) - \mathcal{N}_{FA}^{0} \}~ \cap ~\text{non-local}-\mathbb{FAC} \neq \emptyset
\end{equation}
Since, every neighborhood of $\mathcal{N}_{FA}^{0}$  contains infinitely many points of $\text{non-local}-\mathbb{FAC}$, where $\emptyset$ is the null set.
Next, we construct a sequence $\{\mathcal{N}_{FA}^r\}$ of distinct fidelity annihilating channels  that converges to $\mathcal{N}_{FA}^{0} $ as follows:
\begin{align}
    & \mathcal{N}_{FA}^1 \in B_1(\mathcal{N}_{FA}^{0})~ \cap ~\text{non-local}-\mathbb{FAC} , \hspace{0.15cm} \mathcal{N}_{FA}^1 \neq \mathcal{N}_{FA}^{0}\\ \nonumber
     &  \mathcal{N}_{FA}^2 \in  B_{\frac{1}{2}}(\mathcal{N}_{FA}^{0}) ~\cap~ \text{non-local}-\mathbb{FAC} ,  \hspace{0.15cm} \mathcal{N}_{FA}^2 \neq \mathcal{N}_{FA}^{0},  \mathcal{N}_{FA}^1 \\ \nonumber
      &   ........ \in ........  \\ \nonumber
       &   ........ \in ........  \\ \nonumber
         &  \mathcal{N}_{FA}^r \in  B_{\frac{1}{r}}(\mathcal{N}_{FA}^{0}) ~\cap~ \text{non-local}-\mathbb{FAC} , \hspace{0.15cm}  \mathcal{N}_{FA}^r \neq  \mathcal{N}_{FA}^{0},...,\mathcal{N}_{FA}^{r-1}
   \end{align}
 Now, consider another sequence $\{ \mathcal{N}_{FA}^r ({\rho}_B) = {\chi^r_B} \}$ for arbitrary density operator ${\rho}_B \in  \mathcal{D}({\mathcal{H}_B})$. Since, $\mathcal{N}_{FA}^{r} \in $ $\text{non-local}-\mathbb{FAC}$, it follows that 
\begin{equation}
     F(\mathcal{N}_{FA}^r ({\rho}_B)) = F (\chi^r_B) \leq \frac{1}{d`}
\end{equation}
It is known that the set of states $\mathcal{D}_{\mathbb{FAC}}(\mathcal{H}_{B})$, consisting of states whose fidelity of entanglement is at most $\frac{1}{d`}$ forms a closed set \cite{ganguly2011entanglement}, hence for $ \chi^r_B \in \mathcal{D}_{\mathbb{FAC}}(\mathcal{H}_{B})$, all limit points of $\chi^r_B$ must also belong to $\mathcal{D}_{\mathbb{FAC}}(\mathcal{H}_{B})$. Now, 
\begin{align}
   & \lim_{r\to\infty}  \mathcal{N}_{FA}^r ({\rho}_B) \rightarrow  \mathcal{N}_{FA}^{0} ({\rho}_B)  \\ \nonumber
   & \implies \lim_{r\to\infty} \chi^r_B = \chi^0_B   \hspace{1cm} [\text{where,} \hspace{0.2cm}  \mathcal{N}_{FA}^0 ({\rho}_B) = {\chi^0_B}]
\end{align}
Therefore, 
$\chi^0_B \in  \mathcal{D}_{\mathbb{FAC}}(\mathcal{H}_{B})$, which implies 
\begin{equation}
    F(\mathcal{N}_{FA}^0 ({\rho}_B)) \leq \frac{1}{d`}, \hspace{0.3cm} \text{for arbitrary ${\rho}_B \in  \mathcal{D}({\mathcal{H}_B})$}
\end{equation}
i.e., $\mathcal{N}_{FA}^0 \in $ $\text{non-local}-\mathbb{FAC}$. Since $\mathcal{N}_{FA}^0$ was chosen as an arbitrary limit point of $\text{non-local}-\mathbb{FAC}$, it follows that $\text{non-local}-\mathbb{FAC}$ contains all of its limit points. Therefore, the set $\text{non-local}-\mathbb{FAC}$ is closed. 
\par Furthermore, it is well known that the completely bounded trace norm of quantum channels is equal to 1 \cite{watrous2018theory}. Thus, all channels in the set $\text{non-local}-\mathbb{FAC}$ are bounded maps. Hence, the set $\text{non-local}-\mathbb{FAC}$ is both closed and bounded, i.e., compact.\renewcommand{\qedsymbol}{}
\end{proof}
\subsubsection{\textit{\textbf{Properties of non-local fidelity annihilating channels}}} Below, we examine some structural properties of non-local fidelity annihilating channels.

\begin{property}\label{tgan1}
    If $\mathcal{N}_{FA}^{1}$, $\mathcal{N}_{FA}^{2} \in \text{non-local}-\mathbb{FAC}$ then $\mathcal{N}_{FA}^{1}\circ\mathcal{N}_{FA}^{2}\in \text{non-local}-\mathbb{FAC}$.
\end{property}
\begin{proof}
    Since,~~ $\mathcal{N}_{FA}^{1},~\mathcal{N}_{FA}^{2} \in \text{non-local}-\mathbb{FAC}$, then for any $\rho_{B}\in\mathcal{D}(\mathcal{H}_{B})$ we have
    \begin{equation}
        F(\sigma_{B}^{1}) \leq 1/d` ~~\text{and }~~ F(\sigma_{B}^{2}) \leq 1/d` ~~ \text{where}~~ \sigma_{B}^{1} = \mathcal{N}_{FA}^{1}(\rho_{B}),~~ \sigma_{B}^{2} = \mathcal{N}_{FA}^{2}(\rho_{B})
    \end{equation}
    which further implies
    \begin{equation}\label{sca1}
        \mathcal{N}_{FA}^{1}[\mathcal{D}(\mathcal{H}_{B})] \subset \mathcal{D}_{\mathbb{FAC}}(\mathcal{H}_{B}) ~~\text{and}~~\mathcal{N}_{FA}^{2}[\mathcal{D}(\mathcal{H}_{B})] \subset \mathcal{D}_{\mathbb{FAC}}(\mathcal{H}_{B})
    \end{equation}
    Now, consider the series composition, $\mathcal{N}_{FA}^{1}\circ\mathcal{N}_{FA}^{2}:$
    \begin{equation}\label{sca2}
        \mathcal{N}_{FA}^{1}\circ\mathcal{N}_{FA}^{2} (\rho_{B})~=~ \mathcal{N}_{FA}^{1}[\mathcal{N}_{FA}^{2}(\rho_{B})]
    \end{equation}
    Using Eq.(\ref{sca1}) and Eq.(\ref{sca2}), we obtain
    \begin{equation}
        \mathcal{N}_{FA}^{1}\circ \mathcal{N}_{FA}^{2}(\mathcal{D}(\mathcal{H}_{B}))~=~ \mathcal{N}_{FA}^{1}[\mathcal{N}_{FA}^{2}(\mathcal{D}(\mathcal{H}_{B}))]\subset \mathcal{D}_{\mathbb{FAC}}(\mathcal{H}_{B})
    \end{equation}
    Hence, $\mathcal{N}_{FA}^{1}\circ\mathcal{N}_{FA}^{2} \in \text{non-local}-\mathbb{FAC}.$
    \renewcommand{\qedsymbol}{}

\end{proof}

\begin{property}\label{tgan2}
    If $\mathcal{N}_{FA} \in \text{non-local}-\mathbb{FAC}$ and $\mathcal{T}$ are any arbitrary channels, then $\mathcal{N}_{FA}\circ\mathcal{T}\in \text{non-local}-\mathbb{FAC}$.
\end{property}
\begin{proof}
    Since,~~ $\mathcal{N}_{FA}\in \text{non-local}-\mathbb{FAC}$, ~then for any $\rho_{B}\in\mathcal{D}(\mathcal{H}_{B})$ we have
    \begin{equation}
        F(\sigma_{B}) \leq 1/d`~~ \text{where}~~ \sigma_{B} = \mathcal{N}_{FA}(\rho_{B})
    \end{equation}
    which further implies
    \begin{equation}\label{sca3}
        \mathcal{N}_{FA}[\mathcal{D}(\mathcal{H}_{B})] \subset \mathcal{D}_{\mathbb{FAC}}(\mathcal{H}_{B})
    \end{equation}
    Now, consider series composition, $\mathcal{N}_{FA}\circ\mathcal{T}:$
    \begin{equation}\label{sca4}
        \mathcal{N}_{FA}\circ\mathcal{T}(\rho_{B})~=~ \mathcal{N}_{FA}[\mathcal{T}(\rho_{B})]
    \end{equation}
    Using Eq.(\ref{sca3}) and Eq.(\ref{sca4}), we obtain
    \begin{equation}
        \mathcal{N}_{FA}\circ\mathcal{T}(\mathcal{D}(\mathcal{H}_{B}))~=~ \mathcal{N}_{FA}[\mathcal{T}(\mathcal{D}(\mathcal{H}_{B}))]\subset \mathcal{D}_{\mathbb{FAC}}(\mathcal{H}_{B})
    \end{equation}
    Hence, $\mathcal{N}_{FA}\circ\mathcal{T} \in \text{non-local}-\mathbb{FAC}.$\\

    \renewcommand{\qedsymbol}{}
\end{proof}

\begin{property}
   If $\mathcal{N}_{FA} \in \text{non-local}-\mathbb{FAC}$ annihilates the fidelity of entanglement of all pure states, then it annihilates the fidelity of entanglement of all mixed states.
\end{property}
\begin{proof}
    Since,~~ $\mathcal{N}_{FA}\in \text{non-local}-\mathbb{FAC}$, and annihilate the fidelity of entanglement of pure states, then we have 
    \begin{equation}\label{pse1}
        F(\sigma_{i}) <1/d` ~~\text{where}~~ \sigma_{i} = \mathcal{N}_{FA}(\rho_{i})~~\text{and}~~\rho_{i} = |\psi_{i}\rangle\langle\psi_{i}|~~\text{is pure state.}
    \end{equation}
     For any mixed state $\rho_{B}$, we have 
     \begin{equation}\label{pse2}
         \rho_{B}=\sum_{i}p_{i}\rho_{i}~,~~~~~\sum_{i}p_{i}=1
     \end{equation}
     Now, $\mathcal{N}_{FA}$ acting on input state $\rho_{B}$, we have 
     \begin{equation}\label{pse3}
         \sigma_{B} = \mathcal{N}_{FA}(\rho_{B}) = \mathcal{N}_{FA}(\sum_{i}p_{i}\rho_{i})= \sum_{i}p_{i}\mathcal{N}_{FA}(\rho_{i})
     \end{equation}
     Using Eq.(\ref{pse1}), Eq.(\ref{pse3}), and result from \cite{ganguly2011entanglement}, we obtain
     \begin{equation}
         F(\sigma_{B}) = F(\sum_{i}p_{i}\mathcal{N}_{FA}(\rho_{i})) = \sum_{i}p_{i}F(\sigma_{i}) < 1/d`
     \end{equation}
     Hence, we can say $\mathcal{N}_{FA}$ annihilate the fidelity of entanglement of all mixed states.
     \renewcommand{\qedsymbol}{}
\end{proof}

\subsubsection{\textit{\textbf{Necessary and sufficient condition for the detection of non-local fidelity of entanglement
annihilating channels}}}
Here, we provide the necessary and sufficient condition for the identification of non-local fidelity annihilating channels.
\begin{theorem}
    Let, $\mathcal{N}_{FA}$ be a quantum channel acting on a bipartite subsystem B. Then $\mathcal{N}_{FA} \in$ $\text{non-local}-\mathbb{FAC}$ iff  $\text{Tr} (\mathcal{W} \mathcal{N}_{FA} (\rho_B)) $ is non-negative, for all $\rho_{B} \in  \mathcal{D}(\mathcal{H}_{B})$ and for every teleportation witness $(\mathcal{W})$ satisfying the condition in Eq.~(\ref{witness}).
\end{theorem}
\textit{\textbf{Proof.}}  Let, $\mathcal{N}_{FA} \in$ $\text{non-local}-\mathbb{FAC}$. Therefore, $ F({\mathcal{N}_{FA}} ({\rho}_B)) \le \frac{1}{d`}$ for all  ${\rho}_B \in \mathcal{D}({\mathcal{H}_B})$. Our goal is to show that,  $ \text{Tr} (\mathcal{W} \mathcal{N}_{FA} (\rho_B)) $ is non-negative, for all $\rho_{B} \in  \mathcal{D}(\mathcal{H}_{B})$ and for all teleportation witness $\mathcal{W}$. Now, for arbitrary $\rho_B \in  \mathcal{D}({\mathcal{H}_B})$, 
   \begin{equation}
      \text{Tr} (\mathcal{W} \mathcal{N}_{FA} (\rho_B))  =   \text{Tr} (\mathcal{W}  \sigma_B) \hspace{0.2cm} \text{where}, \hspace{0.2cm}  \mathcal{N}_{FA} ({\rho}_B) = {\sigma}_B
    \end{equation}
      Since, $\mathcal{N}_{FA} \in$ $\text{non-local}-\mathbb{FAC}$, therefore  $ F({\sigma}_B) \le \frac{1}{d`}$. Moreover, since $\mathcal{W}$ is a teleportation witness, hence it preserves non-negativity when acting on any state with fidelity of entanglement less than or equal to $\frac{1}{d`}$ \cite{ganguly2011entanglement}. Hence $ \text{Tr} (\mathcal{W} \mathcal{N}_{FA} (\rho_B)) \ge 0$. Thus for any ${\rho}_B \in \mathcal{D}({\mathcal{H}_B})$, we get $ \text{Tr} (\mathcal{W} \mathcal{N}_{FA} (\rho_B)) \ge 0$, which proves that $\text{Tr} (\mathcal{W} \mathcal{N}_{FA} (\rho_B))$ is non-negative for all $\rho_{B} \in  \mathcal{D}(\mathcal{H}_{AB})$ and for all teleportation witness $\mathcal{W}$.\\

        Conversely, we want to show that if  $\text{Tr} (\mathcal{W} \mathcal{N}_{FA} (\rho_B))$ is non-negative, for all $\rho_{B} \in  \mathcal{D}(\mathcal{H}_{AB})$ and for all teleportation witness $\mathcal{W}$ , then $\mathcal{N}_{FA} \in$ $\text{non-local}-\mathbb{FAC}$ . For arbitrary $\rho_B \in  \mathcal{D}({\mathcal{H}_B})$, consider
     \begin{equation}
         \text{Tr} (\mathcal{W} \mathcal{N}_{FA} (\rho_B))  =   \text{Tr} (\mathcal{W}  \sigma_B),  \hspace{0.2cm} \text{where}, \hspace{0.2cm} \mathcal{N}_{FA} ({\rho}_B) = {\sigma}_B
     \end{equation} 
      By assumption  $ \text{Tr} (\mathcal{W} \mathcal{N}_{FA} (\rho_B)) \ge 0$, which implies that $\text{Tr} (\mathcal{W} (\sigma_B)) \ge 0$, for all teleportation witness $\mathcal{W}$. Hence, $F(\sigma_B) \le \frac{1}{d`}$  \cite{ganguly2011entanglement}, which implies   $F({\mathcal{N}_{FA}} ({\rho}_B)) \le 1/d`$ for all  ${\rho}_B \in \mathcal{D}({\mathcal{H}_B})$ i.e., $\mathcal{N}_{FA} \in$ $\text{non-local}-\mathbb{FAC}$. However, we note here that this condition is not very operational as it needs to be checked for all teleportation witnesses. \\
      
      We now present an illustration of $2$-local fidelity annihilating channel, through an example from depolarizing channels.


\subsubsection{\texorpdfstring{\textit{\textbf{$2$-local fidelity annihilating channel}}}{2-local fidelity annihilating channel}}
We consider the depolarizing channel ${\mathcal{N}_d`}$ given as

\begin{align}\label{dclfa}
    \mathcal{N}_d`(*)=p(*)+(1-p)  \text{Tr}(*)  \frac{\mathbb{I}_{d`}}{d`}
\end{align}
where $p \in [0,1]$. The $2$-local depolarizing channel acts as follows:
\begin{equation}\label{out}
    \begin{split}
   & \omega_{out} = ({\mathcal{N}}_d` \otimes {\mathcal{N}}_d`) (\omega_{B_1 B_2}) \\
   & =    p^2 \omega_{B_1 B_2} + (1-p)^2 (\frac{\mathbb{I}_{d`}}{d`} \otimes \frac{\mathbb{I}_{d`}}{d`}) + p(1-p) (\omega_{B_1} \otimes \frac{\mathbb{I}_{d`}}{d`} + \frac{\mathbb{I}_{d`}}{d`} \otimes \omega_{B_2}) \\ 
   \end{split}
\end{equation}
where, $\omega_{B_1} = \text{Tr}_{B_2}(\omega_{B_1 B_2})$ and $\omega_{B_2} = \text{Tr}_{B_1}(\omega_{B_1 B_2 })$. To prove that $ {\mathcal{N}_d`}$ is a  $2$-local fidelity annihilating channel, it is necessary to verify this for all input states $\omega_{B_1 B_2}$. Since the set of states whose fidelity of entanglement is less than or equal to $\frac{1}{d`}$ forms a convex set \cite{ganguly2011entanglement}, it is sufficient to consider only pure states. Now, any pure state can always be expressed in its Schmidt decomposition form as follows. :
\begin{equation}\label{pst}
    \ket{\psi} = \sum_{j=0}^{d`-1} \sqrt{q_j} \ket{\phi_j} \otimes \ket{\widetilde{\phi}_j} 
\end{equation}
where $\ket{\phi_j}$ and $\ket{\widetilde{\phi}_j}$ form orthonormal bases for the two subsystems on Bob's side. 

\begin{itemize}
    \item Let us first consider the qubit depolarizing channel, i.e., putting $d`=2$ in Eq.\eqref{dclfa}. By implementing  ${\mathcal{N}}_2 \otimes {\mathcal{N}}_2$ on $\omega_{B_1 B_2}$, we get
\begin{equation}\label{outmat}
    \omega_{out} =
    \begin{bmatrix}
        \dfrac{(1-p)^{2}+4pq_{0}}{4} & 0 & 0 & \dfrac{4p^{2}\sqrt{q_{0}q_{1}}}{4} \\[6pt]
        0 & \dfrac{1-p^2}{4} & 0 & 0 \\[6pt]
        0 & 0 & \dfrac{1-p^2}{4} & 0 \\[6pt]
        \dfrac{4p^{2}\sqrt{q_{0}q_{1}}}{4} & 0 & 0 & \dfrac{(1-p)^{2}+4pq_{1}}{4}
    \end{bmatrix}
\end{equation}
where $p\in[0,1]$, $q_{0}\in[0,1]$ and $q_{0}=1-q_{1}$. Now to conclude that the qubit depolarizing channel is $2$-local $\mathbb{FAC}$ channel, we have to show that $F(\omega_{out})\leq1/2$ for some range of $p$. From \cite{horodecki1997inseparable}, we know that the fidelity of entanglement for the two-qubit state is $\frac{1}{4}[1+||T||_{1}]$ where $||T||_{1}=\text{Tr}|T|=\text{Tr}\sqrt{TT^{\dagger}}$ and $T$ is the correlation matrix. Therefore, the fidelity of entanglement of the output state $\omega_{out}$ is given as
\begin{align}
    F(\omega_{out})=\frac{1}{4}[1+p^{2}+4p^{2}\sqrt{q_{0}(1-q_{0})}]
\end{align}
If $F(\omega_{\text{out}}) \leq \frac{1}{2}$, then $\frac{1}{4}[1+p^{2}+4p^{2}\sqrt{q_{0}(1-q_{0})}] \leq \frac{1}{2}$ which simplifies to $p^{2}(1+4\sqrt{q_{0}(1-q_{0})})\leq 1$. 
Now the term $p^{2}(1+4\sqrt{q_{0}(1-q_{0})})$ is maximum if the coefficient of $p^2$ is maximum. The maximum value of $1+4\sqrt{q_{0}(1-q_{0})}$ is $3$ at $q_{0}=\frac{1}{2}$. Therefore, $F(\omega_{\text{out}}) \leq \frac{1}{2}$ if $p\leq \frac{1}{\sqrt{3}}(\thickapprox 0.57735)$ for all $q_{0}$. Therefore, the qubit depolarizing channel belongs to $2$-local $\mathbb{FAC}$ for $p\leq 0.57735$. Additionally, as established in \cite{moravvcikova2010entanglement}, the qubit depolarizing channel also belongs to $2$-local $\mathbb{EA}$ for $p\leq 0.57735$. Thus, the channel satisfies the criteria for belonging to both  $2$-local $\mathbb{FAC}$ and $2$-local $\mathbb{EA}$ in the above range of $p$.



\item Now, we consider the qutrit depolarizing channel i.e., putting $d`=3$ in \eqref{dclfa} and provide the precise range in which this channel is not $2$-locally fidelity annihilating channel. 
The action of the qutrit depolarizing channel is defined as 
\begin{equation} \label{depolarizing}
    {\mathcal{N}}_3 (*) = p (*) + (1-p)  \text{Tr}(*) \hspace{0.1cm} \frac{\mathbb{I}_{3}}{3}
\end{equation}
where $p \in [0,1]$. 

By implementing  ${\mathcal{N}_{3}} \otimes {\mathcal{N}_{3}}$ on $\omega_{B_1 B_2}$, we get
\begin{equation} 
\begin{split}
  \omega_{out} &   = \\
& \begin{bmatrix} 
	\frac{(p^2+2p)q_0}{3}+ t &0& 0 &0 &p^2 \sqrt{q_0 q_1}&0 &0&0&p^2 \sqrt{q_0 q_2}\\[0.1cm]
    0& 	s_1+t &0&0&0&0&0&0 &0\\[0.1cm]
    0&0& 	s_2+t	&0&0&0&0&0&0\\[0.1cm]
     0&0& 0	&s_1+t &0&0&0&0&0\\[0.1cm]
     p^2 \sqrt{q_0 q_1}&0&0& 0&\frac{(p^2+2p)q_1}{3}+ t &0&0&0&p^2 \sqrt{q_1 q_2}\\[0.1cm]
    	 0 &0&0&0&	0&s_3+t &0&0&0 \\[0.1cm]
    0& 0&0&0&0&0&s_2+t&0&0\\[0.1cm]
    0& 0&0&0&0&0&0&s_3+t&0\\[0.1cm]
  p^2 \sqrt{q_0 q_2}&0&0&0&p^2 \sqrt{q_2 q_1}&0&0&0& \frac{(p^2+2p)q_2}{3}+ t\\[0.1cm] \nonumber
	\end{bmatrix}
 \end{split}
	\end{equation}
where, $t=\frac{(1-p)^2}{9}$, $s_1=\frac{{p(1-p)}(q_0 +q_1)}{3}$, $s_2=\frac{{p(1-p)}(q_0 +q_2)}{3}$ and $s_3=\frac{{p(1-p)}(q_1 +q_2)}{3}$.\\

It is difficult to tell whether ${\mathcal{N}}_3$ is a $2$-local fidelity annihilating channel, as computation of the fidelity of entanglement for arbitrary density matrices is challenging. However, using a witness-based approach, we determine below the precise parameter region where the channel ${\mathcal{N}}_3$ \textit{is not a} $2$-local fidelity annihilating channel.

To evaluate the fidelity of entanglement of the output state $\omega_{\text{out}}$, we employ the teleportation witness $( \mathcal{W})$ defined in Eq.(\ref{tw}) for $d=3$.
 If we get $ \text{Tr}( \mathcal{W} \omega_{out}) < 0$, then we can conclude that the channel $\mathcal{N}_3$ is not a $2$-local fidelity annihilating channel. The minimum value of $ \text{Tr} (\mathcal{W}  \omega_{out})$ is $\frac{2-8p^2}{9}$, which is achieved for $q_0 = q_1=q_2= \frac{1}{3}$. It is straightforward to verify that for $p>\frac{1}{2}$,  $\frac{2-8p^2}{9}$ is negative, i.e., $ \text{Tr} (\mathcal{W}  \omega_{out}) < 0$. Therefore, the qutrit depolarizing channel is not a $2$-local fidelity annihilating channel for $p > \frac{1}{2}$.

\end{itemize}

\subsection{\texorpdfstring{Relationship between $\mathbb{FBC}$ and $\mathbb{FAC}$ with Other Resource-Breaking Quantum Channels}{Relationship between FBC and FAC with Other Resource-Breaking Quantum Channels}}
Here, we explore the relationship between $\mathbb{FBC}$ and $2$-local $\mathbb{FAC}$ with other resource breaking quantum channels, specifically with entanglement annihilating channels ($\mathbb{EA}$), negative conditional entropy breaking channels ($\mathbb{NCEBC}$) and negative conditional entropy annihilating channels ($\mathbb{NCEAC}$).\\

\textit{\textbf{Relation between $\mathbb{FBC}$ and $2$-local $\mathbb{FAC}$ }}: Let $\mathcal{N}\in \mathbb{FBC}$ be an arbitrary channel then for any state $\rho_{B}$, where $B$ is bipartite system, we have 

\begin{align} \label{2la}
    (\mathcal{N}\otimes\mathcal{N})(\rho_{B})&= (\mathcal{I}\otimes\mathcal{N}) (\mathcal{N}\otimes\mathcal{I})(\rho_{B}) \nonumber \\
    &= (\mathcal{I}\otimes\mathcal{N})(\sigma^{1}_{B}) \nonumber \\
    &= \sigma^{2}_{B}
\end{align}
We thus have $F(\sigma^{2}_{B})\leq 1/d`$ (since $\mathcal{N}\in \mathbb{FBC}$) which implies that $\mathcal{N}\in$ $2$-local $\mathbb{FAC}$. Thus, we can say that 
\begin{equation}
 \mathbb{FBC} \subseteq 2\text{-local}~\mathbb{FAC}  
\end{equation}

\textbf{\textit{Relation between $2$-local $\mathbb{FAC}$ and $2$-local $\mathbb{EA}$}}~: 
 A single particle channel $\mathcal{N}_{EA}$ belongs to $2$-local $\mathbb{EA}$ if applying $\mathcal{N}_{EA}\otimes \mathcal{N}_{EA}$ on a bipartite subsystem destroys entanglement within that subsystem \cite{moravvcikova2010entanglement}. As a consequence, the fidelity of entanglement of the output state after applying $\mathcal{N}_{EA}\otimes \mathcal{N}_{EA}$ is less than or equal to $1/d$, which implies that all $2$-local entanglement annihilating channels are necessarily $2$-local fidelity annihilating channels. Hence the following relation must hold:
\begin{equation}
   2\text{-local} \hspace{0.2cm}\mathbb{EA} \subseteq 2\text{-local} \hspace{0.2cm} \mathbb{FAC}
\end{equation}


{\textit{\textbf{Relation between $2$-local $\mathbb{FAC}$ and $2$-local $\mathbb{NCEAC}$}}} : Here, we examine the relation between $2$-local $\mathbb{FAC}$ and $2$-local $\mathbb{NCEAC}$. Note that a $2$-local negative conditional entropy annihilating channel acts on a bipartite system $B$ and destroys
the negative conditional entropy within the system $B$ \cite{srinidhi2024quantum}. To explore the connection between these two channels, we consider a qubit depolarizing channel $\mathcal{N}_2$ as defined in Eq.(\ref{dclfa}). Now, to find the exact parameter region where the qubit depolarizing channel belongs to $2$-local $\mathbb{NCEAC}$, we have to consider the negative conditional entropy for all input states $\omega_{B_1 B_2}$, similarly to Eq.\eqref{out}.  Since the set of states with nonnegative conditional entropy forms a convex set \cite{vempati2021witnessing}, it is sufficient to consider only the pure states defined in Eq. (\ref{pst}). Now, after applying the qubit depolarizing on $\omega_{B_1 B_2}$, we obtain $\omega_{out}$ as defined in Eq.\eqref{outmat}. The conditional entropy of $\omega_{out}$ is given by 
\begin{align}
    S(B_1|B_2) &= -\lambda_{1}\log_{2}\lambda_{1} 
    -\lambda_{2}\log_{2}\lambda_{2} 
    -\lambda_{3}\log_{2}\lambda_{3} \notag \\[2ex]
    &\quad -\lambda_{4}\log_{2}\lambda_{4} 
    +\lambda_{5}\log_{2}\lambda_{5}  
    + \lambda_{6}\log_{2}\lambda_{6}
\end{align}
where  
\begin{equation*}
    \begin{aligned}
        \lambda_{1} &= \frac{1 - p^2}{4}, ~~~~~~~~~~~~~~~ \lambda_{2}= \frac{1 - p^2}{4},\\
        \lambda_{3}& = \frac{1 + p^2 - 2\sigma_{1}}{4}, ~~~~~~\lambda_{4} = \frac{1 + p^2 + 2\sigma_{1}}{4},\\
        \lambda_{5} &= \frac{1 - p + 2pq_{0}}{2},~~~~~~~~ \lambda_{6} = \frac{1 + p - 2pq_{0}}{2} , 
    \end{aligned}
\end{equation*}
with
\begin{equation*}
    \sigma_{1} = \sqrt{p^2 - 4 p^2 q_{0} + 4 p^4 q_{0} + 4 p^2 q_{0}^2 - 4 p^4 q_{0}^2}.
\end{equation*}

It can be checked that the conditional entropy $S(B_1|B_2)$ of $\omega_{out}$ is nonnegative for $p \leq 0.86465$. Therefore, the qubit depolarizing channel $\mathcal{N}_{2}$ belongs to $2$-local $\mathbb{NCEAC}$ for  $p \leq 0.86465$. Also, as we have proved earlier, the qubit depolarizing channel $\mathcal{N}_2$ belongs to $2$-local $\mathbb{FAC}$ for $p\leq 0.57735$. Hence, for $ p \leq 0.57735$, the qubit depolarizing channel belongs to both $2$-local $\mathbb{FAC}$ and $2$-local $\mathbb{NCEAC}$ i.e., there exists at least one channel that belongs to both $2$-local $\mathbb{NCEAC}$  and $2$-local $\mathbb{FAC}$ which implies
\begin{equation}
    2\text{-local}~\mathbb{FAC} \cap 2\text{-local}~\mathbb{NCEAC} \neq \emptyset
\end{equation}
Again, this channel does not belong to $2$-local $\mathbb{FAC}$ but belongs to $2$-local $\mathbb{NCEAC}$ for $0.57735< p \leq 0.86465 $. This indicates that there exists at least one channel that does not belong to $2$-local $\mathbb{FAC}$ but belongs to $2$-local $\mathbb{NCEAC}$. Hence, we have the following relation: 
\begin{equation}
 2\text{-local}~\mathbb{NCEAC}   \not \subset~  2\text{-local}~\mathbb{FAC}
\end{equation}
\textit{\textbf{Relation between $\mathbb{FBC}$ and $\mathbb{NCEBC}$}} : 
 Before analyzing the relationship between $\mathbb{FBC}$ and $\mathbb{NCEBC}$, we first examine the conditions under which the qubit depolarizing channel belongs to $\mathbb{FBC}$. The action of the qubit depolarizing channel is given by
\begin{align}\label{FBdclfa}
\mathcal{N}_2(\rho) = p \rho + (1 -  p)\text{Tr}(\rho) \frac{\mathbb{I}_{2}}{2}
\end{align}
where $0 \leq p \leq 1$. Now the action of channel $\mathcal{I}\otimes \mathcal{N}_{2}$ on input state $\omega_{12}$ is as follows: 
\begin{equation}\label{FBout}
    \begin{split}
   & \omega_{out} = (\mathcal{I} \otimes {\mathcal{N}}_2) (\omega_{12}) \\
   & =   p\omega_{12} + (1-p)(\omega_{1}\otimes \frac{\mathbb{I}_{2}}{2}) \\ 
   \end{split}
\end{equation}
where, $\omega_{1} = \text{Tr}_{2}(\omega_{12})$. To prove that qubit-depolarizing channel ${\mathcal{N}_2}$ is a fidelity breaking channel, it is necessary to verify this for all input states $\omega_{12}$. Since the set of states whose fidelity of entanglement is less than or equal to $\frac{1}{2}$ forms a convex set \cite{ganguly2011entanglement}, it is sufficient to consider only pure states.
Now, any pure state $|\psi\rangle$ can always be expressed in its Schmidt decomposition form as follows :
\begin{equation}\label{FBpst}
    \ket{\psi} = \sum_{j=0}^{1} \sqrt{q_j} \ket{\phi_j} \otimes \ket{\widetilde{\phi}_j} 
\end{equation}
where $\ket{\phi_j}$ and $\ket{\widetilde{\phi}_j}$ form orthonormal bases for the two subsystems on Bob's side.
The output state $\omega_{out}$ is given as
\begin{equation}\label{FBoutmat}
    \omega_{out} =
    \begin{bmatrix}
        \dfrac{2pq_{0}+(1-p)q_{0}}{2} & 0 & 0 & p\sqrt{q_{0}q_{1}}\\[6pt]
        0 & \dfrac{(1-p)q_{0}}{2} & 0 & 0 \\[6pt]
        0 & 0 & \dfrac{(1-p)q_{1}}{2} & 0 \\[6pt]
        p\sqrt{q_{0}q_{1}} & 0 & 0 & \dfrac{2pq_{1}+(1-p)q_{1}}{2}
    \end{bmatrix}
\end{equation}
where $q_{0}\in[0,1]$ and $q_{0}=1-q_{1}$. Now to conclude that the qubit depolarizing channel belongs to $\mathbb{FBC}$, we have to show that $F(\omega_{out})\leq1/2$ for a certain range of the parameter $p$. From \cite{horodecki1997inseparable}, we know that the fidelity of entanglement for the two-qubit state is $\frac{1}{4}[1+||T||_{1}]$ where $||T||_{1}=\text{Tr}|T|=\text{Tr}\sqrt{TT^{\dagger}}$ and $T$ is the correlation matrix. Therefore, the fidelity of entanglement of the output state $\omega_{out}$ is given as
\begin{align}
    F(\omega_{out})=\frac{1}{4}[1+p+4p\sqrt{q_{0}(1-q_{0})}]
\end{align}
If $F(\omega_{\text{out}}) \leq \frac{1}{2}$, then $\frac{1}{4}[1+p+4p\sqrt{q_{0}(1-q_{0})}] \leq \frac{1}{2}$ which simplifies to $p(1+4\sqrt{q_{0}(1-q_{0})})\leq 1$. Now the term $p(1+4\sqrt{q_{0}(1-q_{0})})$ is maximum if the coefficient of $p$ is maximum. The maximum value of $1+4\sqrt{q_{0}(1-q_{0})}$ is $3$ at $q_{0}=\frac{1}{2}$. Therefore, $F(\omega_{out}) \leq \frac{1}{2}$ if $p\leq \frac{1}{3}(\thickapprox 0.33333)$ for all $q_{0}$. Therefore, the qubit depolarizing channel belongs to $\mathbb{FBC}$ for $p\leq 0.33333$.\\

We now examine the conditions under which the channel $\mathcal{N}_2$ belongs to $\mathbb{NCEBC}$. To establish this, we compute the conditional entropy of the output state and identify the precise range in which it remains non-negative. Any arbitrary two-qubit pure entangled state can be written as follows:
\begin{align}
    |\psi\rangle = \cos\alpha|00\rangle +\sin\alpha|11\rangle
\end{align}
where $\alpha \in [0,\pi]$. The conditional entropy of the output state $\sigma_{AB} = \ket{\psi} \bra{\psi}$ is given by
\begin{align}
    S(A|B) &= -\lambda_{1}\log_{2}\lambda_{1} 
    -\lambda_{2}\log_{2}\lambda_{2} 
    -\lambda_{3}\log_{2}\lambda_{3} \notag\\[2ex]
    &\quad -\lambda_{4}\log_{2}\lambda_{4} 
    +\lambda_{5}\log_{2}\lambda_{5}  
    + \lambda_{6}\log_{2}\lambda_{6}
\end{align}
where  
\begin{equation*}
    \begin{aligned}
        \lambda_{1} &= \dfrac{1 - p}{2}\cos^2\alpha, ~~~~~ \lambda_{2}= \dfrac{1 - p}{2}\sin^2\alpha,\\
        \lambda_{3}& = \dfrac{2+2p-\sigma_{1}}{8}, ~~~~~~\lambda_{4} = \dfrac{2+2p+\sigma_{1}}{8},\\
        \lambda_{5} &= \dfrac{1 + p \cos2\alpha}{2},~~~~~~ \lambda_{6} = \dfrac{1 - p \cos2\alpha}{2}, 
    \end{aligned}
\end{equation*}
with
\begin{equation*}
    \sigma_{1} = \sqrt{2 + 4p + 10p^2 + 2\cos4\alpha + 4p\cos4\alpha-6p^2\cos4\alpha}.
\end{equation*}

According to \cite{muhuri2023information}, evaluating the output corresponding to a maximally entangled input state is sufficient to determine whether an unital channel is a negative conditional entropy breaking channel. We obtain the conditional entropy is maximized for $\alpha=\frac{\pi}{4}$, and is non-negative for $p\leq 0.747614$. Hence, the qubit depolarizing channel $\mathcal{N}_2$ belongs to $\mathbb{NCEBC}$ for $p\leq 0.747614$.\\

Therefore, we can conclude that for $p\leq0.333333$, the qubit depolarizing channel $\mathcal{N}_{2}$ belongs to both $\mathbb{FBC}$ and $\mathbb{NCEBC}$. Hence, there exists at least one channel that belongs to both $\mathbb{FBC}$ and $\mathbb{NCEBC}$. This confirms that the intersection of these two sets is non-empty, i.e.,
\begin{equation}
    \mathbb{FBC}~\cap~ \mathbb{NCEBC} \neq\emptyset
\end{equation} 
Moreover the qubit depolarizing channel is an element of $\mathbb{NCEBC}$ but does not belong to $\mathbb{FBC}$ for $0.333333 < p \leq 0.747614$. This illustrates the existence of a channel that lies within $\mathbb{NCEBC}$ yet falls outside of $\mathbb{FBC}$. Hence,
\begin{equation}
    \mathbb{NCEBC}~ \not \subset~ \mathbb{FBC}
\end{equation}
     
Motivated by the connections between fidelity breaking and annihilating channels with negative conditional entropy breaking and annihilating channels, in the next section, we explore how the fidelity of entanglement and various entropies are related in the realm of quantum states.

		\section{Quantum Entropies and fidelity of entanglement for two-qubit states} \label{IV} 
        In \cite{kumar2024quantum}, the relationship between fidelity of entanglement and quantum conditional entropy was studied for a specific class of states—those with maximally mixed marginals. Here, we extend that investigation to arbitrary two-qubit states, and further examine the connections not only with quantum conditional entropy, but also with other quantum entropy measures, including the min-Rényi entropy and its associated conditional form.
        \subsection{General two-qubit state}
        For general two-qubit state $\rho_{AB}$ from Eq.(\ref{2qgs}), the correlation matrix defined in Eq.(\ref{cmat}) is given by
\begin{center}
    $ T =
    \begin{bmatrix}
      t_{11} & t_{12} & t_{13}\\[6pt]
      t_{21} & t_{22} & t_{23}\\[6pt]
      t_{31} & t_{32} & t_{33}
     \end{bmatrix}
     $
\end{center}
The singular values of correlation matrix $T$ are $\sigma_{i}$, $i=1,2,3$, and without loss of generality we take $\sigma_{1}\geq\sigma_{2}\geq\sigma_{3}\geq 0.$ The trace norm $||T||_{1}$ and Frobenius norm $||T||_{2}$ are respectively given by 
\begin{align}{\label{tn}}
    ||T||_{1} = \text{Tr}(\sqrt{T^{\dagger}T}) = \sum_{i=1}^3\sigma_{i}
\end{align}
\begin{align}{\label{fn}}
    ||T||_{2} = [\text{Tr}(T^{\dagger}T)]^{\frac{1}{2}} = (\sum_{i=1}^3\sigma_{i}^{2})^{\frac{1}{2}}
\end{align}
Now, using Eq.(\ref{tn}) and Eq.(\ref{fn}), we have the following lemma.
\begin{lemma}{\label{lm2}}
    $||T||_{1}>1$ if and only if $||T||_{2}^{2} > 1-R$, where $R=2(\sigma_{1}\sigma_{2}+\sigma_{1}\sigma_{3}+\sigma_{2}\sigma_{3})$
\end{lemma}
Using Eq.(\ref{fgh}) and Lemma(\ref{lm2}), we have the following result.
\begin{theorem}{\label{tf}}
    The fidelity of entanglement of the general two-qubit state $\rho_{AB}$ is strictly greater than half, i.e., $F(\rho_{AB})>1/2$ iff $||T||_{2}^{2} > 1-R$, where $R=2(\sigma_{1}\sigma_{2}+\sigma_{1}\sigma_{3}+\sigma_{2}\sigma_{3})$
\end{theorem}

\subsubsection{{\textbf{Relation with R\'enyi entropy}}} 
In this subsection, we investigate the relation between R\'enyi entropy and the fidelity of entanglement of the general two-qubit state. The R\'enyi $2$-entropy of the general two-qubit state $\rho_{AB}$ is given by \begin{align}
      S_{2}(AB) = \log_{2}[\frac{4}{1+||\vec{a}||^{2}+||\vec{b}||^{2}+||T||_{2}^{2}}]
  \end{align} 
  where $||\cdot||$ is the Euclidean norm and $||T||_{2}^{2}=\text{Tr}(T^{\dagger}T)$.
    \begin{theorem}
        The fidelity of entanglement of the general two-qubit state $\rho_{AB}$ is strictly greater than half, i.e., $F(\rho_{AB})>1/2$ iff $S_{2}(AB)<\log_{2}\Gamma$ where $\Gamma > 0$ and $\Gamma=\frac{4}{2+||\vec{a}||^2+||\vec{b}||^2-R}$
    \end{theorem}
    \renewcommand{\proofname}{\textbf{Proof}} 

\begin{proof}
    \smallskip
    Let $F(\rho_{AB}) > 1/2$
    
    \smallskip
    $\Longleftrightarrow ||T||_{2}^{2} > 1 - R$, \ using Theorem (\ref{tf})
    
    \smallskip
    $\Longleftrightarrow \frac{1 + ||\vec{a}||^{2} + ||\vec{b}||^{2} + ||T||_{2}^{2}}{4} > \frac{2 + ||\vec{a}||^{2} + ||\vec{b}||^{2} - R}{4}$
    
    \smallskip
    $\Longleftrightarrow \log_{2} [ \frac{4}{1 + ||\vec{a}||^{2} + ||\vec{b}||^{2} + ||T||_{2}^{2}} ] < \log_{2} [ \frac{4}{2 + ||\vec{a}||^{2} + ||\vec{b}||^{2} - R} ]$
    
    \smallskip
    $\Longleftrightarrow S_{2}(AB) < \log [ \frac{4}{2 + ||\vec{a}||^2 + ||\vec{b}||^2 - R} ].$
    
    \renewcommand{\qedsymbol}{}
\end{proof}

Conditional R\'enyi $2$-entropy of general two-qubit state $\rho_{AB}$ is given by 
\begin{align}
    S_{2}(A|B)=\log_{2}[\frac{2+2||\vec{b}||^{2}}{1+||\vec{a}||^{2}+||\vec{b}||^{2}+||T||_{2}^{2}}]
\end{align}
  \begin{theorem}
      The fidelity of entanglement of general two-qubit state $\rho_{AB}$ is strictly greater than half iff $S_{2}(A|B)<\log_{2}\Delta$ where $\Delta>0$ and $\Delta=\frac{2+ 2 ||\vec{b}||^2}{2+||\vec{a}||^2+||\vec{b}||^2-R}$.
  \end{theorem}

  \begin{proof}
    \smallskip
    Let $F(\rho_{AB}) > 1/2$
    
    \smallskip
    $\Longleftrightarrow ||T||_{2}^{2} > 1 - R$, \ using Theorem (\ref{tf})
    
    \smallskip
    $\Longleftrightarrow \frac{1}{1 + ||\vec{a}||^{2} + ||\vec{b}||^{2} + ||T||_{2}^{2}} < \frac{1}{2 + ||\vec{a}||^{2} + ||\vec{b}||^{2} - R}$
    
    \smallskip
    $\Longleftrightarrow \log_{2} [ \frac{2 + 2||\vec{b}||^{2}}{1 + ||\vec{a}||^{2} + ||\vec{b}||^{2} + ||T||_{2}^{2}} ] < \log_{2} [ \frac{2 + 2||\vec{b}||^{2}}{2 + ||\vec{a}||^{2} + ||\vec{b}||^{2} - R} ]$
    
    \smallskip
    $\Longleftrightarrow S_{2}(A|B) < \log [ \frac{2 + 2||\vec{b}||^{2}}{2 + ||\vec{a}||^2 + ||\vec{b}||^2 - R} ].$
    
    \renewcommand{\qedsymbol}{}
\end{proof}

Now, we see the relation between the Min-R\'enyi entropy and the fidelity of entanglement. Since Min-R\'enyi entropy is given by \begin{align}\label{met}
      S_{\infty}(AB) = -\log_{2}||\rho_{AB}||_{O}
  \end{align}
  where $||\cdot||_{O}$ denotes the operator norm. In addition, the operator norm $||\rho_{AB}||_{O}$ is the largest eigenvalue of $\rho_{AB}$. From \cite{rui2010estimation}, for any two-qubit state $\rho_{AB}$, the fidelity of entanglement $F(\rho_{AB})$ satisfies $F(\rho_{AB}) \leq \text{max}_{i}\{\lambda_{i}\}$ where $\lambda_{i}$ are eigenvalues of state $\rho_{AB}$. Since, the operator norm $||\rho_{AB}||_{O}$ is the largest eigenvalue of state $\rho_{AB}$. So, we have the following inequality
  \begin{align}{\label{fev}}
      F(\rho_{AB}) \leq ||\rho_{AB}||_{O} 
  \end{align}
  where $||\rho_{AB}||_{O}$ is the largest eigenvalue of $\rho_{AB}$. Using Eq.(\ref{met}) and Eq.(\ref{fev}), we have: $F(\rho_{AB})\leq ||\rho_{AB}||_{O} \Leftrightarrow \log_{2}F(\rho_{AB}) \leq \log_{2}||\rho_{AB}||_{O} \Leftrightarrow S_{\infty}(AB) \leq -\log_{2}F(\rho_{AB}).$ So, the following result is obtained, 
  \begin{theorem}\label{tfme}
      For any general two-qubit state $\rho_{AB}$, the fidelity of entanglement and the Min-R\'enyi entropy satisfy $S_{\infty}(AB) \leq -\log_{2}F(\rho_{AB}).$
  \end{theorem}
  The conditional Min-R\'enyi entropy of the two-qubit state is given by 
  \begin{align}{\label{cmet}}
      S_{\infty}(A|B) = \log_{2}\frac{||\rho_{B}||_{O}}{||\rho_{AB}||_{O}}
  \end{align}
  where $||\rho_{AB}||_{O}$ and $||\rho_{B}||_{O}$ are the largest eigenvalues of the state $\rho_{AB}$ and $\rho_{B}$, respectively. Using Eq.(\ref{cmet}) and Eq.(\ref{fev}), we have: $ F(\rho_{AB})\leq ||\rho_{AB}||_{O}\Leftrightarrow \frac{F(\rho_{AB})}{||\rho_{B}||_{O}} \leq \frac{||\rho_{AB}||_{O}}{||\rho_{B}||_{O}} \Leftrightarrow \log_{2}\frac{||\rho_{B}||_{O}}{||\rho_{AB}||_{O}} \leq \log_{2}\frac{||\rho_{B}||_{O}}{F(\rho_{AB})} \Leftrightarrow S_{\infty}(A|B) \leq \log_{2}\frac{||\rho_{B}||_{O}}{F(\rho_{AB})}$. So, we obtained the following result.
  \begin{theorem}\label{tfcme}
      For any general two-qubit state $\rho_{AB}$, the fidelity of entanglement and conditional Min-R\'enyi entropy satisfy $ S_{\infty}(A|B) \leq -\log_{2}\frac{F(\rho_{AB})}{||\rho_{B}||_{O}}.$
  \end{theorem}

  Now, we investigate the relation between fidelity of entanglement strictly greater than $\frac{1}{2}$, min-R\'enyi entropy, and conditional min-R\'enyi entropy for a general two-qubit state. If $F(\rho_{AB})>1/2$, then we have the following inequality
  \begin{align}\label{mef}
      -\log_{2}F(\rho_{AB})<1
  \end{align}
  From Theorem (\ref{tfme}) and Eq.(\ref{mef}), we obtain the following result
  \begin{theorem}
      If the fidelity of entanglement of the general two-qubit state is greater than half, i.e., $F(\rho_{AB})>1/2$, then its Min-R\'enyi entropy satisfies $S_{\infty}(AB)<1.$
  \end{theorem}
  If $F(\rho_{AB})>1/2$, then we have the following inequality
  \begin{align}\label{cmef}
      -\log_{2}\frac{F(\rho_{AB})}{||\rho_{B}||_{O}} < \log_{2}2||\rho_{B}||_{O}
  \end{align}
  From Theorem(\ref{tfcme}) and Eq.(\ref{cmef}), we obtain the following result,
  \begin{theorem}
      If the fidelity of entanglement of the general two-qubit state is greater than half, i.e., $F(\rho_{AB})>1/2$, then its conditional min-R\'enyi entropy satisfies $S_{\infty}(A|B)<\log_{2}2||\rho_{B}||_{O}.$
  \end{theorem}

		\subsubsection{{\textbf{Relation with Tsallis entropy}}}
   In this subsection, we investigate the relation between Tsallis entropy and the fidelity of entanglement of the general two-qubit state.
  Tsallis $2$-entropy of the general two-qubit state $\rho_{AB}$ is given by 
  \begin{align}
      S_{2}^{\mathbb{T}}(AB) = \frac{3-||\vec{a}||^{2}-||\vec{b}||^{2}-||T||_{2}^{2}}{4}
  \end{align} 
  where $||\cdot||$ is Euclidean norm and $||T||_{2}^{2}=\text{Tr}(T^{\dagger}T)$.
  \begin{theorem}
      The fidelity of entanglement of the general two-qubit state $\rho_{AB}$ is strictly greater than half, i.e., $F(\rho_{AB})>1/2$ iff $S_{2}^{\mathbb{T}}(AB)<\eta$ where $\eta=\frac{2-||\vec{a}||^2-||\vec{b}||^2+R}{4}.$
  \end{theorem}

  \begin{proof}
    \smallskip
    Let $F(\rho_{AB}) > 1/2$ 
    
    \smallskip
    $\Longleftrightarrow ||T||_{2}^{2} > 1 - R$, \ using Theorem (\ref{tf})
    
    \smallskip
    $\Longleftrightarrow \frac{1 + ||\vec{a}||^{2} + ||\vec{b}||^{2} + ||T||_{2}^{2}}{4} > \frac{2 + ||\vec{a}||^{2} + ||\vec{b}||^{2} - R}{4}$
    
    \smallskip
    $\Longleftrightarrow \frac{3 - ||\vec{a}||^{2} - ||\vec{b}||^{2} - ||T||_{2}^{2}}{4} < \frac{2 - ||\vec{a}||^{2} - ||\vec{b}||^{2} + R}{4}$
    
    \smallskip
    $\Longleftrightarrow S_{2}^{\mathbb{T}}(AB) < \frac{2 - ||\vec{a}||^2 - ||\vec{b}||^2 + R}{4}.$
    
    \renewcommand{\qedsymbol}{}
\end{proof}

Conditional Tsallis $2$-entropy of general two-qubit state $\rho_{AB}$ is given by 
\begin{align}
    S_{2}^{\mathbb{T}}(A|B)=\frac{1-||\vec{a}||^{2}+||\vec{b}||^{2}-||T||_{2}^{2}}{4}
\end{align}
  \begin{theorem}
      The fidelity of entanglement of the general two-qubit state $\rho_{AB}$ is strictly greater than half, i.e., $F(\rho_{AB})>1/2$ iff $S_{2}^{\mathbb{T}}(A|B)<\Lambda$ where $\Lambda=\frac{||\vec{b}||^{2}-||\vec{a}||^{2}+R}{4}$
  \end{theorem}

  \begin{proof}
    \smallskip
    Let $F(\rho_{AB}) > 1/2$ 
    
    \smallskip
    $\Longleftrightarrow ||T||_{2}^{2} > 1 - R$, \ using Theorem (\ref{tf})
    
    \smallskip
    $\Longleftrightarrow \frac{1 + ||\vec{a}||^{2} + ||\vec{b}||^{2} + ||T||_{2}^{2}}{4} > \frac{2 + ||\vec{a}||^{2} + ||\vec{b}||^{2} - R}{4}$
    
    \smallskip
    $\Longleftrightarrow \frac{1 - ||\vec{a}||^{2} + ||\vec{b}||^{2} - ||T||_{2}^{2}}{4} < \frac{||\vec{b}||^{2} - ||\vec{a}||^{2} + R}{4}$
    
    \smallskip
    $\Longleftrightarrow S_{2}^{\mathbb{T}}(A|B) < \frac{||\vec{b}||^{2} - ||\vec{a}||^{2} + R}{4}$
    
    \renewcommand{\qedsymbol}{}
\end{proof}

  \subsection{Two-qubit Weyl state}
  Correlation matrix of two-qubit Weyl state is given by $T=\text{diagonal}\{t_{11}, t_{22}, t_{33}\}$. So, $||T||_{1}=\text{Tr}(|T|)=|t_{11}| + |t_{22}|+ |t_{33}|.$ Here, we investigate the relation between fidelity of entanglement and various entropies.
  \subsubsection{\textbf{Relation with R\'enyi entropy}}
  R\'eny $2$-entropy of two-qubit Weyl state is given by 
  \begin{align}{\label{rew1}}
      S_{2}(AB) = -\log_{2}[\frac{1+t_{11}^{2}+t_{22}^{2}+t_{33}^{2}}{4}]
  \end{align}
  Using Eq.(\ref{fgh}) and Eq.(\ref{rew1}), we have the following observation.
  \begin{observation}
      The fidelity of entanglement of two-qubit Weyl state $\rho_{AB}$ is greater than half, i.e., $F(\rho_{AB})>1/2$ iff $S_{2}(AB)<\log_{2}\frac{2}{1-\Omega}$ where $\Omega=|t_{11}||t_{22}|+|t_{11}||t_{33}|+|t_{22}||t_{33}|$ and $0<\Omega<1$. 
  \end{observation}
  Conditional R\'enyi $2$-entropy of two-qubit Weyl state is given by
  \begin{align}{\label{rew2}}
      S_{2}(A|B) = 1 - \log_{2}(1+t_{11}^{2}+t_{22}^{2}+t_{33}^{2})
  \end{align}
  Using Eq.(\ref{fgh}) and Eq.(\ref{rew2}), we have the following observation.
  \begin{observation}
      The fidelity of entanglement of two-qubit Weyl state $\rho_{AB}$ is greater than half, i.e., $F(\rho_{AB})>1/2$ iff $S_{2}(A|B)<\log_{2}\frac{1}{1-\Omega}$ where $\Omega=|t_{11}||t_{22}|+|t_{11}||t_{33}|+|t_{22}||t_{33}|$ and $0<\Omega<1$. 
  \end{observation}
  Now, we will see the relation between min-R\'enyi entropy and faithful entanglement. Since min-R\'enyi entropy is given by 
  \begin{align}
      S_{\infty}(AB) = -\log_{2}||\rho_{AB}||_{O}
  \end{align}
  where $||\cdot||_{O}$ denotes the operator norm, i.e., $||\rho_{AB}||_{O}$ is the largest eigenvalue of $\rho_{AB}$. In \cite{quan2016steering}, for $|t_{3}|\leq t_{2}\leq t_{1}\leq 1$ the spectrum of two-qubit Weyl state $\rho_{AB}$ in non decreasing order is given as $\frac{1}{4}(1-t_{1}-t_{2}-t_{3}), \frac{1}{4}(1-t_{1}+t_{2}+t_{3}), \frac{1}{4}(1+t_{1}-t_{2}+t_{3}), \frac{1}{4}(1+t_{1}+t_{2}-t_{3})$. The smallest and largest eigenvalues of $\rho_{AB}$ is given by $\lambda_{\text{min}}=\frac{1}{4}(1-t_{1}-t_{2}-t_{3})$ and $\lambda_{\text{max}}=\frac{1}{4}(1+t_{1}+t_{2}-t_{3}).$ From \cite{riccardi2021exploring,guhne2021geometry}, we know that the fidelity of entanglement of the two-qubit Weyl state is greater than half if and only if its largest eigenvalue is strictly greater than $\frac{1}{2}$. Thus, we have the following observation for the two-qubit Weyl state
  \begin{observation}
      The fidelity of entanglement of the two-qubit Weyl state $\rho_{AB}$ is greater than half if and only if $S_{\infty}(AB)<1.$
  \end{observation}
  \begin{proof}
    \smallskip
    Let $F(\rho_{AB}) > 1/2$
    $\Longleftrightarrow \lambda_{\text{max}} > \frac{1}{2}$
    
    \smallskip
    $\Longleftrightarrow \log_{2}(\lambda_{\text{max}}) > \log_{2} \left( \frac{1}{2} \right) = -1$
    
    \smallskip
    $\Longleftrightarrow -\log_{2}(\lambda_{\text{max}}) < 1$
    $\Longleftrightarrow -\log_{2} ||\rho_{AB}||_{O} < 1$
    
    \smallskip
    $\Longleftrightarrow$ \quad $ S_{\infty}(AB) < 1.$
    
    \renewcommand{\qedsymbol}{}
\end{proof}

  Conditional min-R\'enyi entropy of two-qubit Weyl state $\rho_{AB}$ is given by $S_{\infty}(A|B) = -\log_{2}(2\lambda_{\text{max}})$. 
  \begin{observation}
      The fidelity of entanglement of the two-qubit Weyl state is greater than half if and only if $S_{\infty}(A|B)<0.$
  \end{observation}

  \begin{proof}
    \smallskip
    Let $F(\rho_{AB}) > 1/2$
    $\Longleftrightarrow \lambda_{\text{max}} > \frac{1}{2}$
    
    \smallskip
    $\Longleftrightarrow \log_{2}(\lambda_{\text{max}}) > \log_{2} \left( \frac{1}{2} \right) = -\log_{2} 2$
    
    \smallskip
    $\Longleftrightarrow -\log_{2} (2\lambda_{\text{max}}) < 0$ $\Longleftrightarrow S_{\infty}(A|B) < 0$
    
    \renewcommand{\qedsymbol}{}
\end{proof}
 
  \subsubsection{\textbf{Relation with Tsallis entropy}}
  Tsallis $2-$entropy of two-qubit Weyl $\rho_{AB}$ state is given as
  \begin{equation}\label{ewt}
      S_{2}^{\mathbb{T}}(AB)=\frac{3-t_{11}^{2}-t_{22}^{2}-t_{33}^{2}}{4}
  \end{equation}
  Using Eq.(\ref{fgh}) and Eq.(\ref{ewt}), we have the following observation.
  \begin{observation}
      The fidelity of entanglement of the two-qubit Weyl state $\rho_{AB}$ is greater than half if and only if $S_{2}^{\mathbb{T}}(AB)<\frac{1+\eta}{2}$ where $\eta=|t_{11}||t_{22}|+|t_{11}||t_{33}|+|t_{22}||t_{33}|.$
  \end{observation}
  Conditional Tsallis $2-$entropy of two-qubit Weyl state is given as
  \begin{equation}\label{ewct}
      S_{2}^{\mathbb{T}}(A|B)=\frac{1-t_{11}^{2}-t_{22}^{2}-t_{33}^{2}}{4}
  \end{equation}
  Using Eq.(\ref{fgh}) and Eq.(\ref{ewct}), we have the following observation.
  \begin{observation}
      The fidelity of entanglement of the two-qubit Weyl state $\rho_{AB}$ is greater than half if and only if $S_{2}^{\mathbb{T}}(A|B)<\frac{\eta}{2}$ where $\eta=|t_{11}||t_{22}|+|t_{11}||t_{33}|+|t_{22}||t_{33}|.$
  \end{observation}

   \section{\texorpdfstring{Relation between relative Entropy and fidelity of entanglement in $d \otimes d$ systems}{Relation between relative Entropy and fidelity of entanglement in d x d systems}} \label{V}
   
   In this section, we explore the relation between relative entropy and fidelity of entanglement. The fidelity of entanglement of state $\rho_{AB}$ is given as 
   \begin{align}\label{mrf1}
        F(\rho_{AB})=\max\limits_{ U }\langle\phi|(U^{\dagger}\otimes \mathbb{I}) \rho_{AB} (U\otimes \mathbb{I}) |\phi\rangle
   \end{align}
   where $|\phi\rangle=\frac{1}{\sqrt{d}}\sum\limits_{i=0}^{d-1}|ii\rangle$, is maximally entangled state.
   Further simplifying Eq.(\ref{mrf1}), we get
   \begin{align}\label{mrf2}
        F(\rho_{AB})=\max\limits_{ U } \text{Tr}[\rho_{AB}(U\otimes \mathbb{I} )\rho_{\phi}(U^{\dagger}\otimes \mathbb{I})]
   \end{align}
   where $\rho_{\phi}=|\phi\rangle\langle\phi|.$
   Letting $(U\otimes \mathbb{I} )\rho_{\phi}(U^{\dagger}\otimes \mathbb{I})= \sigma_{AB}$, we consider the relative entropy \cite{ruskai2002inequalities, vedral1997quantifying, wilde2013quantum, nielsen2010quantum} denoted by $D(\sigma_{AB}||\rho_{AB})$, given as
   \begin{align}
       D(\sigma_{AB}||\rho_{AB})&= \text{Tr}[\sigma_{AB}(\log\sigma_{AB}-\log\rho_{AB})]\nonumber\\
       &=-\text{Tr}[\sigma_{AB}\log\rho_{AB}]~;~ (\text{Since, $\sigma_{AB}$ is pure}\Rightarrow~\text{Tr}[\sigma_{AB}\log\sigma_{AB}]=0) \nonumber\\
       &=-\text{Tr}[(U\otimes \mathbb{I} )\rho_{\phi}(U^{\dagger}\otimes \mathbb{I})\log\rho_{AB}]\nonumber\\
       &= - \text{Tr}[\log\rho_{AB}(U\otimes \mathbb{I} )\rho_{\phi}(U^{\dagger}\otimes \mathbb{I})])\nonumber
   \end{align}
   Now, we define a quantity $R(\rho_{AB})$ as
   \begin{align}\label{mrf3}
       R(\rho_{AB})=\max\limits_{ U }\{- \text{Tr}[\log\rho_{AB}(U\otimes \mathbb{I} )\rho_{\phi}(U^{\dagger}\otimes \mathbb{I})]\}
   \end{align}
   From Eq.(\ref{mrf2}), we have
   \begin{align}\label{mrf4}
       -F(\rho_{AB})=\max\limits_{ U } \{- \text{Tr}[\rho_{AB}(U\otimes \mathbb{I} )\rho_{\phi}(U^{\dagger}\otimes \mathbb{I})]\}
   \end{align}
   \begin{theorem}
       For density matrix $\rho_{AB}$, $R(\rho_{AB})$ and fidelity of entanglement $F(\rho_{AB})$ satisfy $R(\rho_{AB})\geq -F(\rho_{AB})$.
   \end{theorem}
   \begin{proof}
       Consider a function $f(x)=-\log(x)$ that is a convex function for all $x>0$ \& $X$ is a random variable then, by Jensen's inequality \cite{durrett2019probability}, we have
       \begin{align}\label{j1}
           f[E(X)]\leq E[f(X)]
       \end{align}
       In quantum scenarios Eq.(\ref{j1}) becomes
       \begin{align}\label{j2}
           f[\text{Tr}(\mathbb{AB})]\leq\text{Tr}[f(\mathbb{A})\mathbb{B}]
       \end{align}
       where $\mathbb{A}$ is density matrix and $\mathbb{B}$ is positive semidefinite operator. Further simplifying Eq.(\ref{j2}), and putting the $f(x)=-\log(x)$ and due to its operator convexity, we obtain
       \begin{align}\label{j3}
           -\text{Tr}[\log(\mathbb{A})\mathbb{B}] \geq -\log[\text{Tr}(\mathbb{AB})]
       \end{align}
       Now, we will check the relation between $-\log[\text{Tr}(\mathbb{AB})]$ and $-\text{Tr}(\mathbb{AB})$.
       Since, $-\log(x)\geq -x+1$ for all $0<x\leq1$. The $\text{Tr}(\mathbb{AB})$ is expectation value such that $0<\text{Tr}(\mathbb{AB})\leq 1$ where $B\geq 0$ and $B$ is normalized. So, we have 
       \begin{align}\label{j4}
           -\log[\text{Tr}(\mathbb{AB})] &\geq  -\text{Tr}(\mathbb{AB}) +1 \nonumber\\
           &\geq -\text{Tr}(\mathbb{AB})
       \end{align}
       Using Eq.(\ref{j3}) and Eq.(\ref{j4}), we have 
       \begin{align}\label{j5}
           -\text{Tr}[\log(\mathbb{A})\mathbb{B}] \geq -\text{Tr}(\mathbb{AB})
       \end{align}
       Consider density matrices $\rho_{AB}$ and $(U\otimes \mathbb{I} )\rho_{\phi}(U^{\dagger}\otimes \mathbb{I})$ such that $\mathbb{A}=\rho_{AB}$ and $\mathbb{B}=(U\otimes \mathbb{I} )\rho_{\phi}(U^{\dagger}\otimes \mathbb{I}).$ \\
       Maximizing over unitary operator $U$ on the both side of Eq.(\ref{j5}), we obtain
       \begin{align}
           \max\limits_{ U }\{- \text{Tr}[\log\rho_{AB}(U\otimes \mathbb{I} )\rho_{\phi}(U^{\dagger}\otimes \mathbb{I})]\} \geq \max\limits_{ U } \{- \text{Tr}[\rho_{AB}(U\otimes \mathbb{I} )\rho_{\phi}(U^{\dagger}\otimes \mathbb{I})]\}
       \end{align}
       Hence, $R(\rho_{AB})\geq -F(\rho_{AB})$.
       \renewcommand{\qedsymbol}{}
   \end{proof}

\section{Conclusion}\label{VI}
  Certifying entanglement is a fundamental area of study, as entangled states serve as a crucial resource for various quantum information processing tasks \cite{ekert1991quantum,branciard2012one,pati2000minimum,hillery1999quantum,cleve1999share,bandyopadhyay2000teleportation,bennett1993teleporting}. Fidelity of entanglement and different measures of quantum entropy play a pivotal role in assessing and certifying the presence of entanglement in a given quantum state. In this work, we aim to examine the interplay between these two indicators of entanglement, both in state and channel regimes. First, we present a comprehensive analysis and characterization of channels that break the fidelity of entanglement of bipartite composite systems. We then introduce the concept of fidelity annihilating channels, which act on a bipartite system $B$ and reduce the fidelity of entanglement within the system $B$. We investigate the topological characterization and several fundamental properties of fidelity annihilating channels. Additionally, we provide examples of these channels, with a specific focus on analyzing the depolarizing channel. we then explore the relationship between $\mathbb{FBC}$ and $2$-locally $\mathbb{FAC}$ with other resource breaking quantum channels, specifically with entanglement annihilating channels ($\mathbb{EA}$), negative conditional entropy breaking channels ($\mathbb{NCEBC}$ ) and negative conditional entropy annihilating channels ($\mathbb{NCEAC}$).
  \par Transitioning from the channel level to the state level, we establish a connection between quantum entropies and the fidelity of entanglement for a general two-qubit state. We derive the upper bound for the R\'enyi $2$-entropy and the conditional R\'enyi $2$-entropy, providing a necessary and sufficient condition for a general two-qubit state to exhibit a fidelity of entanglement greater than $\frac{1}{2}$ of a general two-qubit state. We also observe that the min-R\'enyi entropy and conditional min-R\'enyi entropy are both upper-bounded when the state is useful for teleportation. Additionally, we present upper bounds for the Tsallis $2$-entropy and conditional Tsallis $2$-entropy of general two-qubit states, which ensure that the fidelity of entanglement is greater than $\frac{1}{2}$ and vice versa. Finally, we extend our study to explore the relationship between relative entropy and the fidelity of entanglement.
  \par Our study paves the way for several promising research directions. One potential avenue is to investigate the relationship between the fidelity of entanglement and various quantum entropies in higher-dimensional quantum states. Additionally, while we have provided a topological characterization of the fidelity annihilating channels, an important future direction could involve determining their Choi-Kraus representation, offering a deeper structural understanding of these channels.\\

\noindent{\bf{Acknowledgement:}} We thank Mir Alimuddin for valuable insights. B.M. acknowledges the DST INSPIRE fellowship program for financial support. N.G. acknowledges support from SERB(DST) MATRICS grant vide file number MTR/2022/000101. 

\section*{Declaration}
\noindent All the authors contributed equally to the manuscript.\\
{\bf{Data availability:}} Data sharing was not applicable to this article as no data sets were generated or analyzed during the current study.\\
{\bf{Competing interests:}} The authors have no competing interests to declare. All co-authors have seen and agree with the contents of the manuscript, and there is no financial interest to report.

  \bibliography{main}

\begin{thebibliography}{87}%
\makeatletter
\providecommand \@ifxundefined [1]{%
 \@ifx{#1\undefined}
}%
\providecommand \@ifnum [1]{%
 \ifnum #1\expandafter \@firstoftwo
 \else \expandafter \@secondoftwo
 \fi
}%
\providecommand \@ifx [1]{%
 \ifx #1\expandafter \@firstoftwo
 \else \expandafter \@secondoftwo
 \fi
}%
\providecommand \natexlab [1]{#1}%
\providecommand \enquote  [1]{``#1''}%
\providecommand \bibnamefont  [1]{#1}%
\providecommand \bibfnamefont [1]{#1}%
\providecommand \citenamefont [1]{#1}%
\providecommand \href@noop [0]{\@secondoftwo}%
\providecommand \href [0]{\begingroup \@sanitize@url \@href}%
\providecommand \@href[1]{\@@startlink{#1}\@@href}%
\providecommand \@@href[1]{\endgroup#1\@@endlink}%
\providecommand \@sanitize@url [0]{\catcode `\\12\catcode `\$12\catcode `\&12\catcode `\#12\catcode `\^12\catcode `\_12\catcode `\%12\relax}%
\providecommand \@@startlink[1]{}%
\providecommand \@@endlink[0]{}%
\providecommand \url  [0]{\begingroup\@sanitize@url \@url }%
\providecommand \@url [1]{\endgroup\@href {#1}{\urlprefix }}%
\providecommand \urlprefix  [0]{URL }%
\providecommand \Eprint [0]{\href }%
\providecommand \doibase [0]{https://doi.org/}%
\providecommand \selectlanguage [0]{\@gobble}%
\providecommand \bibinfo  [0]{\@secondoftwo}%
\providecommand \bibfield  [0]{\@secondoftwo}%
\providecommand \translation [1]{[#1]}%
\providecommand \BibitemOpen [0]{}%
\providecommand \bibitemStop [0]{}%
\providecommand \bibitemNoStop [0]{.\EOS\space}%
\providecommand \EOS [0]{\spacefactor3000\relax}%
\providecommand \BibitemShut  [1]{\csname bibitem#1\endcsname}%
\let\auto@bib@innerbib\@empty
\bibitem [{\citenamefont {Peres}(1997)}]{peres1997quantum}%
  \BibitemOpen
  \bibfield  {author} {\bibinfo {author} {\bibfnamefont {A.}~\bibnamefont {Peres}},\ }\href@noop {} {\emph {\bibinfo {title} {Quantum theory: concepts and methods}}}\ (\bibinfo  {publisher} {Springer},\ \bibinfo {year} {1997})\BibitemShut {NoStop}%
\bibitem [{\citenamefont {Bohr}(1935)}]{bohr1935can}%
  \BibitemOpen
  \bibfield  {author} {\bibinfo {author} {\bibfnamefont {N.}~\bibnamefont {Bohr}},\ }\bibfield  {title} {\bibinfo {title} {Can quantum-mechanical description of physical reality be considered complete?},\ }\href {https://doi.org/10.1103/PhysRev.48.696} {\bibfield  {journal} {\bibinfo  {journal} {Physical review}\ }\textbf {\bibinfo {volume} {48}},\ \bibinfo {pages} {696} (\bibinfo {year} {1935})}\BibitemShut {NoStop}%
\bibitem [{\citenamefont {Einstein}\ \emph {et~al.}(1935)\citenamefont {Einstein}, \citenamefont {Podolsky},\ and\ \citenamefont {Rosen}}]{einstein1935can}%
  \BibitemOpen
  \bibfield  {author} {\bibinfo {author} {\bibfnamefont {A.}~\bibnamefont {Einstein}}, \bibinfo {author} {\bibfnamefont {B.}~\bibnamefont {Podolsky}},\ and\ \bibinfo {author} {\bibfnamefont {N.}~\bibnamefont {Rosen}},\ }\bibfield  {title} {\bibinfo {title} {Can quantum-mechanical description of physical reality be considered complete?},\ }\href {https://doi.org/10.1103/PhysRev.47.777} {\bibfield  {journal} {\bibinfo  {journal} {Physical review}\ }\textbf {\bibinfo {volume} {47}},\ \bibinfo {pages} {777} (\bibinfo {year} {1935})}\BibitemShut {NoStop}%
\bibitem [{\citenamefont {Horodecki}\ \emph {et~al.}(2009)\citenamefont {Horodecki}, \citenamefont {Horodecki}, \citenamefont {Horodecki},\ and\ \citenamefont {Horodecki}}]{horodecki2009quantum}%
  \BibitemOpen
  \bibfield  {author} {\bibinfo {author} {\bibfnamefont {R.}~\bibnamefont {Horodecki}}, \bibinfo {author} {\bibfnamefont {P.}~\bibnamefont {Horodecki}}, \bibinfo {author} {\bibfnamefont {M.}~\bibnamefont {Horodecki}},\ and\ \bibinfo {author} {\bibfnamefont {K.}~\bibnamefont {Horodecki}},\ }\bibfield  {title} {\bibinfo {title} {Quantum entanglement},\ }\href {https://doi.org/10.1103/RevModPhys.81.865} {\bibfield  {journal} {\bibinfo  {journal} {Reviews of modern physics}\ }\textbf {\bibinfo {volume} {81}},\ \bibinfo {pages} {865} (\bibinfo {year} {2009})}\BibitemShut {NoStop}%
\bibitem [{\citenamefont {Ekert}(1991)}]{ekert1991quantum}%
  \BibitemOpen
  \bibfield  {author} {\bibinfo {author} {\bibfnamefont {A.~K.}\ \bibnamefont {Ekert}},\ }\bibfield  {title} {\bibinfo {title} {Quantum cryptography based on bell’s theorem},\ }\href {https://doi.org/10.1103/PhysRevLett.67.661} {\bibfield  {journal} {\bibinfo  {journal} {Physical review letters}\ }\textbf {\bibinfo {volume} {67}},\ \bibinfo {pages} {661} (\bibinfo {year} {1991})}\BibitemShut {NoStop}%
\bibitem [{\citenamefont {Branciard}\ \emph {et~al.}(2012)\citenamefont {Branciard}, \citenamefont {Cavalcanti}, \citenamefont {Walborn}, \citenamefont {Scarani},\ and\ \citenamefont {Wiseman}}]{branciard2012one}%
  \BibitemOpen
  \bibfield  {author} {\bibinfo {author} {\bibfnamefont {C.}~\bibnamefont {Branciard}}, \bibinfo {author} {\bibfnamefont {E.~G.}\ \bibnamefont {Cavalcanti}}, \bibinfo {author} {\bibfnamefont {S.~P.}\ \bibnamefont {Walborn}}, \bibinfo {author} {\bibfnamefont {V.}~\bibnamefont {Scarani}},\ and\ \bibinfo {author} {\bibfnamefont {H.~M.}\ \bibnamefont {Wiseman}},\ }\bibfield  {title} {\bibinfo {title} {One-sided device-independent quantum key distribution: Security, feasibility, and the connection with steering},\ }\href {https://doi.org/10.1103/PhysRevA.85.010301} {\bibfield  {journal} {\bibinfo  {journal} {Physical Review A}\ }\textbf {\bibinfo {volume} {85}},\ \bibinfo {pages} {010301} (\bibinfo {year} {2012})}\BibitemShut {NoStop}%
\bibitem [{\citenamefont {Pati}(2000)}]{pati2000minimum}%
  \BibitemOpen
  \bibfield  {author} {\bibinfo {author} {\bibfnamefont {A.~K.}\ \bibnamefont {Pati}},\ }\bibfield  {title} {\bibinfo {title} {Minimum classical bit for remote preparation and measurement of a qubit},\ }\href {https://doi.org/10.1103/PhysRevA.63.014302} {\bibfield  {journal} {\bibinfo  {journal} {Physical Review A}\ }\textbf {\bibinfo {volume} {63}},\ \bibinfo {pages} {014302} (\bibinfo {year} {2000})}\BibitemShut {NoStop}%
\bibitem [{\citenamefont {Hillery}\ \emph {et~al.}(1999)\citenamefont {Hillery}, \citenamefont {Bu{\v{z}}ek},\ and\ \citenamefont {Berthiaume}}]{hillery1999quantum}%
  \BibitemOpen
  \bibfield  {author} {\bibinfo {author} {\bibfnamefont {M.}~\bibnamefont {Hillery}}, \bibinfo {author} {\bibfnamefont {V.}~\bibnamefont {Bu{\v{z}}ek}},\ and\ \bibinfo {author} {\bibfnamefont {A.}~\bibnamefont {Berthiaume}},\ }\bibfield  {title} {\bibinfo {title} {Quantum secret sharing},\ }\href {https://doi.org/10.1103/PhysRevA.59.1829} {\bibfield  {journal} {\bibinfo  {journal} {Physical Review A}\ }\textbf {\bibinfo {volume} {59}},\ \bibinfo {pages} {1829} (\bibinfo {year} {1999})}\BibitemShut {NoStop}%
\bibitem [{\citenamefont {Cleve}\ \emph {et~al.}(1999)\citenamefont {Cleve}, \citenamefont {Gottesman},\ and\ \citenamefont {Lo}}]{cleve1999share}%
  \BibitemOpen
  \bibfield  {author} {\bibinfo {author} {\bibfnamefont {R.}~\bibnamefont {Cleve}}, \bibinfo {author} {\bibfnamefont {D.}~\bibnamefont {Gottesman}},\ and\ \bibinfo {author} {\bibfnamefont {H.-K.}\ \bibnamefont {Lo}},\ }\bibfield  {title} {\bibinfo {title} {How to share a quantum secret},\ }\href {https://doi.org/10.1103/PhysRevLett.83.648} {\bibfield  {journal} {\bibinfo  {journal} {Physical review letters}\ }\textbf {\bibinfo {volume} {83}},\ \bibinfo {pages} {648} (\bibinfo {year} {1999})}\BibitemShut {NoStop}%
\bibitem [{\citenamefont {Bandyopadhyay}(2000)}]{bandyopadhyay2000teleportation}%
  \BibitemOpen
  \bibfield  {author} {\bibinfo {author} {\bibfnamefont {S.}~\bibnamefont {Bandyopadhyay}},\ }\bibfield  {title} {\bibinfo {title} {Teleportation and secret sharing with pure entangled states},\ }\href {https://doi.org/10.1103/PhysRevA.62.012308} {\bibfield  {journal} {\bibinfo  {journal} {Physical Review A}\ }\textbf {\bibinfo {volume} {62}},\ \bibinfo {pages} {012308} (\bibinfo {year} {2000})}\BibitemShut {NoStop}%
\bibitem [{\citenamefont {Bennett}\ \emph {et~al.}(1993)\citenamefont {Bennett}, \citenamefont {Brassard}, \citenamefont {Cr{\'e}peau}, \citenamefont {Jozsa}, \citenamefont {Peres},\ and\ \citenamefont {Wootters}}]{bennett1993teleporting}%
  \BibitemOpen
  \bibfield  {author} {\bibinfo {author} {\bibfnamefont {C.~H.}\ \bibnamefont {Bennett}}, \bibinfo {author} {\bibfnamefont {G.}~\bibnamefont {Brassard}}, \bibinfo {author} {\bibfnamefont {C.}~\bibnamefont {Cr{\'e}peau}}, \bibinfo {author} {\bibfnamefont {R.}~\bibnamefont {Jozsa}}, \bibinfo {author} {\bibfnamefont {A.}~\bibnamefont {Peres}},\ and\ \bibinfo {author} {\bibfnamefont {W.~K.}\ \bibnamefont {Wootters}},\ }\bibfield  {title} {\bibinfo {title} {Teleporting an unknown quantum state via dual classical and einstein-podolsky-rosen channels},\ }\href {https://doi.org/10.1103/PhysRevLett.70.1895} {\bibfield  {journal} {\bibinfo  {journal} {Physical review letters}\ }\textbf {\bibinfo {volume} {70}},\ \bibinfo {pages} {1895} (\bibinfo {year} {1993})}\BibitemShut {NoStop}%
\bibitem [{\citenamefont {Lewenstein}\ \emph {et~al.}(2001)\citenamefont {Lewenstein}, \citenamefont {Kraus}, \citenamefont {Horodecki},\ and\ \citenamefont {Cirac}}]{lewenstein2001characterization}%
  \BibitemOpen
  \bibfield  {author} {\bibinfo {author} {\bibfnamefont {M.}~\bibnamefont {Lewenstein}}, \bibinfo {author} {\bibfnamefont {B.}~\bibnamefont {Kraus}}, \bibinfo {author} {\bibfnamefont {P.}~\bibnamefont {Horodecki}},\ and\ \bibinfo {author} {\bibfnamefont {J.}~\bibnamefont {Cirac}},\ }\bibfield  {title} {\bibinfo {title} {Characterization of separable states and entanglement witnesses},\ }\href {https://doi.org/10.1103/PhysRevA.63.044304} {\bibfield  {journal} {\bibinfo  {journal} {Physical Review A}\ }\textbf {\bibinfo {volume} {63}},\ \bibinfo {pages} {044304} (\bibinfo {year} {2001})}\BibitemShut {NoStop}%
\bibitem [{\citenamefont {Terhal}(2000)}]{terhal2000bell}%
  \BibitemOpen
  \bibfield  {author} {\bibinfo {author} {\bibfnamefont {B.~M.}\ \bibnamefont {Terhal}},\ }\bibfield  {title} {\bibinfo {title} {Bell inequalities and the separability criterion},\ }\href {https://doi.org/10.1016/S0375-9601(00)00401-1} {\bibfield  {journal} {\bibinfo  {journal} {Physics Letters A}\ }\textbf {\bibinfo {volume} {271}},\ \bibinfo {pages} {319} (\bibinfo {year} {2000})}\BibitemShut {NoStop}%
\bibitem [{\citenamefont {G{\"u}hne}\ and\ \citenamefont {T{\'o}th}(2009)}]{guhne2009entanglement}%
  \BibitemOpen
  \bibfield  {author} {\bibinfo {author} {\bibfnamefont {O.}~\bibnamefont {G{\"u}hne}}\ and\ \bibinfo {author} {\bibfnamefont {G.}~\bibnamefont {T{\'o}th}},\ }\bibfield  {title} {\bibinfo {title} {Entanglement detection},\ }\href {https://doi.org/10.1016/j.physrep.2009.02.004} {\bibfield  {journal} {\bibinfo  {journal} {Physics Reports}\ }\textbf {\bibinfo {volume} {474}},\ \bibinfo {pages} {1} (\bibinfo {year} {2009})}\BibitemShut {NoStop}%
\bibitem [{\citenamefont {Ganguly}\ and\ \citenamefont {Adhikari}(2009)}]{ganguly2009witness}%
  \BibitemOpen
  \bibfield  {author} {\bibinfo {author} {\bibfnamefont {N.}~\bibnamefont {Ganguly}}\ and\ \bibinfo {author} {\bibfnamefont {S.}~\bibnamefont {Adhikari}},\ }\bibfield  {title} {\bibinfo {title} {Witness for edge states and its characteristics},\ }\href {https://doi.org/10.1103/PhysRevA.80.032331} {\bibfield  {journal} {\bibinfo  {journal} {Physical Review A}\ }\textbf {\bibinfo {volume} {80}},\ \bibinfo {pages} {032331} (\bibinfo {year} {2009})}\BibitemShut {NoStop}%
\bibitem [{\citenamefont {Ganguly}\ \emph {et~al.}(2014)\citenamefont {Ganguly}, \citenamefont {Chatterjee},\ and\ \citenamefont {Majumdar}}]{ganguly2014witness}%
  \BibitemOpen
  \bibfield  {author} {\bibinfo {author} {\bibfnamefont {N.}~\bibnamefont {Ganguly}}, \bibinfo {author} {\bibfnamefont {J.}~\bibnamefont {Chatterjee}},\ and\ \bibinfo {author} {\bibfnamefont {A.}~\bibnamefont {Majumdar}},\ }\bibfield  {title} {\bibinfo {title} {Witness of mixed separable states useful for entanglement creation},\ }\href {https://doi.org/10.1103/PhysRevA.89.052304} {\bibfield  {journal} {\bibinfo  {journal} {Physical Review A}\ }\textbf {\bibinfo {volume} {89}},\ \bibinfo {pages} {052304} (\bibinfo {year} {2014})}\BibitemShut {NoStop}%
\bibitem [{\citenamefont {Ganguly}\ \emph {et~al.}(2013)\citenamefont {Ganguly}, \citenamefont {Adhikari},\ and\ \citenamefont {Majumdar}}]{ganguly2013common}%
  \BibitemOpen
  \bibfield  {author} {\bibinfo {author} {\bibfnamefont {N.}~\bibnamefont {Ganguly}}, \bibinfo {author} {\bibfnamefont {S.}~\bibnamefont {Adhikari}},\ and\ \bibinfo {author} {\bibfnamefont {A.~S.}\ \bibnamefont {Majumdar}},\ }\bibfield  {title} {\bibinfo {title} {Common entanglement witnesses and their characteristics},\ }\href {https://doi.org/10.1007/s11128-012-0386-7} {\bibfield  {journal} {\bibinfo  {journal} {Quantum information processing}\ }\textbf {\bibinfo {volume} {12}},\ \bibinfo {pages} {425} (\bibinfo {year} {2013})}\BibitemShut {NoStop}%
\bibitem [{\citenamefont {Mallick}\ and\ \citenamefont {Nandi}(2024)}]{mallick2024genuine}%
  \BibitemOpen
  \bibfield  {author} {\bibinfo {author} {\bibfnamefont {B.}~\bibnamefont {Mallick}}\ and\ \bibinfo {author} {\bibfnamefont {S.}~\bibnamefont {Nandi}},\ }\bibfield  {title} {\bibinfo {title} {Genuine entanglement detection via projection map in multipartite systems},\ }\href {https://doi.org/10.1088/1402-4896/ad7204} {\bibfield  {journal} {\bibinfo  {journal} {Physica Scripta}\ }\textbf {\bibinfo {volume} {99}},\ \bibinfo {pages} {105116} (\bibinfo {year} {2024})}\BibitemShut {NoStop}%
\bibitem [{\citenamefont {Mallick}\ \emph {et~al.}(2025)\citenamefont {Mallick}, \citenamefont {Maity}, \citenamefont {Ganguly},\ and\ \citenamefont {Majumdar}}]{mallick2025higher}%
  \BibitemOpen
  \bibfield  {author} {\bibinfo {author} {\bibfnamefont {B.}~\bibnamefont {Mallick}}, \bibinfo {author} {\bibfnamefont {A.~G.}\ \bibnamefont {Maity}}, \bibinfo {author} {\bibfnamefont {N.}~\bibnamefont {Ganguly}},\ and\ \bibinfo {author} {\bibfnamefont {A.}~\bibnamefont {Majumdar}},\ }\bibfield  {title} {\bibinfo {title} {Higher-dimensional-entanglement detection and quantum-channel characterization using moments of generalized positive maps},\ }\href {https://doi.org/10.1103/nzrc-8yrt} {\bibfield  {journal} {\bibinfo  {journal} {Physical Review A}\ }\textbf {\bibinfo {volume} {112}},\ \bibinfo {pages} {012416} (\bibinfo {year} {2025})}\BibitemShut {NoStop}%
\bibitem [{\citenamefont {Mukherjee}\ \emph {et~al.}(2025)\citenamefont {Mukherjee}, \citenamefont {Mallick}, \citenamefont {Naik}, \citenamefont {Maity},\ and\ \citenamefont {Majumdar}}]{mukherjee2025efficient}%
  \BibitemOpen
  \bibfield  {author} {\bibinfo {author} {\bibfnamefont {S.}~\bibnamefont {Mukherjee}}, \bibinfo {author} {\bibfnamefont {B.}~\bibnamefont {Mallick}}, \bibinfo {author} {\bibfnamefont {S.~G.}\ \bibnamefont {Naik}}, \bibinfo {author} {\bibfnamefont {A.~G.}\ \bibnamefont {Maity}},\ and\ \bibinfo {author} {\bibfnamefont {A.}~\bibnamefont {Majumdar}},\ }\bibfield  {title} {\bibinfo {title} {Efficient detection of genuine multipartite entanglement using moments of positive maps},\ }\bibfield  {journal} {\bibinfo  {journal} {arXiv preprint arXiv:2506.00162}\ }\href {https://doi.org/10.48550/arXiv.2506.00162} {10.48550/arXiv.2506.00162} (\bibinfo {year} {2025})\BibitemShut {NoStop}%
\bibitem [{\citenamefont {Cerf}\ and\ \citenamefont {Adami}(1999)}]{cerf1999quantum}%
  \BibitemOpen
  \bibfield  {author} {\bibinfo {author} {\bibfnamefont {N.~J.}\ \bibnamefont {Cerf}}\ and\ \bibinfo {author} {\bibfnamefont {C.}~\bibnamefont {Adami}},\ }\bibfield  {title} {\bibinfo {title} {Quantum extension of conditional probability},\ }\href {https://doi.org/10.1103/PhysRevA.60.893} {\bibfield  {journal} {\bibinfo  {journal} {Physical Review A}\ }\textbf {\bibinfo {volume} {60}},\ \bibinfo {pages} {893} (\bibinfo {year} {1999})}\BibitemShut {NoStop}%
\bibitem [{\citenamefont {Horodecki}\ \emph {et~al.}(1999)\citenamefont {Horodecki}, \citenamefont {Horodecki},\ and\ \citenamefont {Horodecki}}]{horodecki1999general}%
  \BibitemOpen
  \bibfield  {author} {\bibinfo {author} {\bibfnamefont {M.}~\bibnamefont {Horodecki}}, \bibinfo {author} {\bibfnamefont {P.}~\bibnamefont {Horodecki}},\ and\ \bibinfo {author} {\bibfnamefont {R.}~\bibnamefont {Horodecki}},\ }\bibfield  {title} {\bibinfo {title} {General teleportation channel, singlet fraction, and quasidistillation},\ }\href {https://doi.org/10.1103/PhysRevA.60.1888} {\bibfield  {journal} {\bibinfo  {journal} {Physical Review A}\ }\textbf {\bibinfo {volume} {60}},\ \bibinfo {pages} {1888} (\bibinfo {year} {1999})}\BibitemShut {NoStop}%
\bibitem [{\citenamefont {R{\'e}nyi}(1961)}]{renyi1961measures}%
  \BibitemOpen
  \bibfield  {author} {\bibinfo {author} {\bibfnamefont {A.}~\bibnamefont {R{\'e}nyi}},\ }\bibfield  {title} {\bibinfo {title} {On measures of entropy and information},\ }in\ \href@noop {} {\emph {\bibinfo {booktitle} {in Proceedings of the fourth Berkeley symposium on mathematical statistics and probability}}},\ Vol.~\bibinfo {volume} {1}\ (\bibinfo {year} {1961})\ p.\ \bibinfo {pages} {547}\BibitemShut {NoStop}%
\bibitem [{\citenamefont {Linden}\ \emph {et~al.}(2013)\citenamefont {Linden}, \citenamefont {Mosonyi},\ and\ \citenamefont {Winter}}]{linden2013structure}%
  \BibitemOpen
  \bibfield  {author} {\bibinfo {author} {\bibfnamefont {N.}~\bibnamefont {Linden}}, \bibinfo {author} {\bibfnamefont {M.}~\bibnamefont {Mosonyi}},\ and\ \bibinfo {author} {\bibfnamefont {A.}~\bibnamefont {Winter}},\ }\bibfield  {title} {\bibinfo {title} {The structure of r{\'e}nyi entropic inequalities},\ }\href {https://doi.org/10.1098/rspa.2012.0737} {\bibfield  {journal} {\bibinfo  {journal} {Proceedings of the Royal Society A: Mathematical, Physical and Engineering Sciences}\ }\textbf {\bibinfo {volume} {469}},\ \bibinfo {pages} {20120737} (\bibinfo {year} {2013})}\BibitemShut {NoStop}%
\bibitem [{\citenamefont {Tsallis}(1988)}]{tsallis1988possible}%
  \BibitemOpen
  \bibfield  {author} {\bibinfo {author} {\bibfnamefont {C.}~\bibnamefont {Tsallis}},\ }\bibfield  {title} {\bibinfo {title} {Possible generalization of boltzmann-gibbs statistics},\ }\href {https://doi.org/10.1007/BF01016429} {\bibfield  {journal} {\bibinfo  {journal} {Journal of statistical physics}\ }\textbf {\bibinfo {volume} {52}},\ \bibinfo {pages} {479} (\bibinfo {year} {1988})}\BibitemShut {NoStop}%
\bibitem [{\citenamefont {M{\"u}ller-Lennert}\ \emph {et~al.}(2013)\citenamefont {M{\"u}ller-Lennert}, \citenamefont {Dupuis}, \citenamefont {Szehr}, \citenamefont {Fehr},\ and\ \citenamefont {Tomamichel}}]{muller2013quantum}%
  \BibitemOpen
  \bibfield  {author} {\bibinfo {author} {\bibfnamefont {M.}~\bibnamefont {M{\"u}ller-Lennert}}, \bibinfo {author} {\bibfnamefont {F.}~\bibnamefont {Dupuis}}, \bibinfo {author} {\bibfnamefont {O.}~\bibnamefont {Szehr}}, \bibinfo {author} {\bibfnamefont {S.}~\bibnamefont {Fehr}},\ and\ \bibinfo {author} {\bibfnamefont {M.}~\bibnamefont {Tomamichel}},\ }\bibfield  {title} {\bibinfo {title} {On quantum r{\'e}nyi entropies: A new generalization and some properties},\ }\href {https://doi.org/10.1063/1.4838856} {\bibfield  {journal} {\bibinfo  {journal} {Journal of Mathematical Physics}\ }\textbf {\bibinfo {volume} {54}},\ \bibinfo {pages} {12} (\bibinfo {year} {2013})}\BibitemShut {NoStop}%
\bibitem [{\citenamefont {Von~Neumann}(2018)}]{von2018mathematical}%
  \BibitemOpen
  \bibfield  {author} {\bibinfo {author} {\bibfnamefont {J.}~\bibnamefont {Von~Neumann}},\ }\href@noop {} {\emph {\bibinfo {title} {Mathematical foundations of quantum mechanics: New edition}}},\ Vol.~\bibinfo {volume} {53}\ (\bibinfo  {publisher} {Princeton university press},\ \bibinfo {year} {2018})\BibitemShut {NoStop}%
\bibitem [{\citenamefont {Cerf}\ and\ \citenamefont {Adami}(1997{\natexlab{a}})}]{cerf1997entropic}%
  \BibitemOpen
  \bibfield  {author} {\bibinfo {author} {\bibfnamefont {N.~J.}\ \bibnamefont {Cerf}}\ and\ \bibinfo {author} {\bibfnamefont {C.}~\bibnamefont {Adami}},\ }\bibfield  {title} {\bibinfo {title} {Entropic bell inequalities},\ }\href {https://doi.org/10.1103/PhysRevA.55.3371} {\bibfield  {journal} {\bibinfo  {journal} {Physical Review A}\ }\textbf {\bibinfo {volume} {55}},\ \bibinfo {pages} {3371} (\bibinfo {year} {1997}{\natexlab{a}})}\BibitemShut {NoStop}%
\bibitem [{\citenamefont {Schneeloch}\ \emph {et~al.}(2013)\citenamefont {Schneeloch}, \citenamefont {Broadbent}, \citenamefont {Walborn}, \citenamefont {Cavalcanti},\ and\ \citenamefont {Howell}}]{schneeloch2013einstein}%
  \BibitemOpen
  \bibfield  {author} {\bibinfo {author} {\bibfnamefont {J.}~\bibnamefont {Schneeloch}}, \bibinfo {author} {\bibfnamefont {C.~J.}\ \bibnamefont {Broadbent}}, \bibinfo {author} {\bibfnamefont {S.~P.}\ \bibnamefont {Walborn}}, \bibinfo {author} {\bibfnamefont {E.~G.}\ \bibnamefont {Cavalcanti}},\ and\ \bibinfo {author} {\bibfnamefont {J.~C.}\ \bibnamefont {Howell}},\ }\bibfield  {title} {\bibinfo {title} {Einstein-podolsky-rosen steering inequalities from entropic uncertainty relations},\ }\href {https://doi.org/10.1103/PhysRevA.87.062103} {\bibfield  {journal} {\bibinfo  {journal} {Physical Review A}\ }\textbf {\bibinfo {volume} {87}},\ \bibinfo {pages} {062103} (\bibinfo {year} {2013})}\BibitemShut {NoStop}%
\bibitem [{\citenamefont {Costa}\ \emph {et~al.}(2018)\citenamefont {Costa}, \citenamefont {Uola},\ and\ \citenamefont {G{\"u}hne}}]{costa2018entropic}%
  \BibitemOpen
  \bibfield  {author} {\bibinfo {author} {\bibfnamefont {A.~C.}\ \bibnamefont {Costa}}, \bibinfo {author} {\bibfnamefont {R.}~\bibnamefont {Uola}},\ and\ \bibinfo {author} {\bibfnamefont {O.}~\bibnamefont {G{\"u}hne}},\ }\bibfield  {title} {\bibinfo {title} {Entropic steering criteria: applications to bipartite and tripartite systems},\ }\href {https://doi.org/10.3390/e20100763} {\bibfield  {journal} {\bibinfo  {journal} {Entropy}\ }\textbf {\bibinfo {volume} {20}},\ \bibinfo {pages} {763} (\bibinfo {year} {2018})}\BibitemShut {NoStop}%
\bibitem [{\citenamefont {Wehner}\ and\ \citenamefont {Winter}(2010)}]{wehner2010entropic}%
  \BibitemOpen
  \bibfield  {author} {\bibinfo {author} {\bibfnamefont {S.}~\bibnamefont {Wehner}}\ and\ \bibinfo {author} {\bibfnamefont {A.}~\bibnamefont {Winter}},\ }\bibfield  {title} {\bibinfo {title} {Entropic uncertainty relations—a survey},\ }\href {https://doi.org/10.1088/1367-2630/12/2/025009} {\bibfield  {journal} {\bibinfo  {journal} {New Journal of Physics}\ }\textbf {\bibinfo {volume} {12}},\ \bibinfo {pages} {025009} (\bibinfo {year} {2010})}\BibitemShut {NoStop}%
\bibitem [{\citenamefont {Cerf}\ and\ \citenamefont {Adami}(1997{\natexlab{b}})}]{cerf1997negative}%
  \BibitemOpen
  \bibfield  {author} {\bibinfo {author} {\bibfnamefont {N.~J.}\ \bibnamefont {Cerf}}\ and\ \bibinfo {author} {\bibfnamefont {C.}~\bibnamefont {Adami}},\ }\bibfield  {title} {\bibinfo {title} {Negative entropy and information in quantum mechanics},\ }\href {https://doi.org/10.1103/PhysRevLett.79.5194} {\bibfield  {journal} {\bibinfo  {journal} {Physical Review Letters}\ }\textbf {\bibinfo {volume} {79}},\ \bibinfo {pages} {5194} (\bibinfo {year} {1997}{\natexlab{b}})}\BibitemShut {NoStop}%
\bibitem [{\citenamefont {Vempati}\ \emph {et~al.}(2021)\citenamefont {Vempati}, \citenamefont {Ganguly}, \citenamefont {Chakrabarty},\ and\ \citenamefont {Pati}}]{vempati2021witnessing}%
  \BibitemOpen
  \bibfield  {author} {\bibinfo {author} {\bibfnamefont {M.}~\bibnamefont {Vempati}}, \bibinfo {author} {\bibfnamefont {N.}~\bibnamefont {Ganguly}}, \bibinfo {author} {\bibfnamefont {I.}~\bibnamefont {Chakrabarty}},\ and\ \bibinfo {author} {\bibfnamefont {A.~K.}\ \bibnamefont {Pati}},\ }\bibfield  {title} {\bibinfo {title} {Witnessing negative conditional entropy},\ }\href {https://doi.org/10.1103/PhysRevA.104.012417} {\bibfield  {journal} {\bibinfo  {journal} {Physical Review A}\ }\textbf {\bibinfo {volume} {104}},\ \bibinfo {pages} {012417} (\bibinfo {year} {2021})}\BibitemShut {NoStop}%
\bibitem [{\citenamefont {Vempati}\ \emph {et~al.}(2022)\citenamefont {Vempati}, \citenamefont {Shah}, \citenamefont {Ganguly},\ and\ \citenamefont {Chakrabarty}}]{vempati2022unital}%
  \BibitemOpen
  \bibfield  {author} {\bibinfo {author} {\bibfnamefont {M.}~\bibnamefont {Vempati}}, \bibinfo {author} {\bibfnamefont {S.}~\bibnamefont {Shah}}, \bibinfo {author} {\bibfnamefont {N.}~\bibnamefont {Ganguly}},\ and\ \bibinfo {author} {\bibfnamefont {I.}~\bibnamefont {Chakrabarty}},\ }\bibfield  {title} {\bibinfo {title} {A-unital operations and quantum conditional entropy},\ }\href {https://doi.org/10.22331/q-2022-02-02-641} {\bibfield  {journal} {\bibinfo  {journal} {Quantum}\ }\textbf {\bibinfo {volume} {6}},\ \bibinfo {pages} {641} (\bibinfo {year} {2022})}\BibitemShut {NoStop}%
\bibitem [{\citenamefont {Brandsen}\ \emph {et~al.}(2021)\citenamefont {Brandsen}, \citenamefont {Geng}, \citenamefont {Wilde},\ and\ \citenamefont {Gour}}]{brandsen2021quantum}%
  \BibitemOpen
  \bibfield  {author} {\bibinfo {author} {\bibfnamefont {S.}~\bibnamefont {Brandsen}}, \bibinfo {author} {\bibfnamefont {I.~J.}\ \bibnamefont {Geng}}, \bibinfo {author} {\bibfnamefont {M.~M.}\ \bibnamefont {Wilde}},\ and\ \bibinfo {author} {\bibfnamefont {G.}~\bibnamefont {Gour}},\ }\bibfield  {title} {\bibinfo {title} {Quantum conditional entropy from information-theoretic principles},\ }\bibfield  {journal} {\bibinfo  {journal} {arXiv preprint arXiv:2110.15330}\ }\href {https://doi.org/abs/2110.15330} {abs/2110.15330} (\bibinfo {year} {2021})\BibitemShut {NoStop}%
\bibitem [{\citenamefont {Patro}\ \emph {et~al.}(2017)\citenamefont {Patro}, \citenamefont {Chakrabarty},\ and\ \citenamefont {Ganguly}}]{patro2017non}%
  \BibitemOpen
  \bibfield  {author} {\bibinfo {author} {\bibfnamefont {S.}~\bibnamefont {Patro}}, \bibinfo {author} {\bibfnamefont {I.}~\bibnamefont {Chakrabarty}},\ and\ \bibinfo {author} {\bibfnamefont {N.}~\bibnamefont {Ganguly}},\ }\bibfield  {title} {\bibinfo {title} {Non-negativity of conditional von neumann entropy and global unitary operations},\ }\href {https://doi.org/10.1103/PhysRevA.96.062102} {\bibfield  {journal} {\bibinfo  {journal} {Physical Review A}\ }\textbf {\bibinfo {volume} {96}},\ \bibinfo {pages} {062102} (\bibinfo {year} {2017})}\BibitemShut {NoStop}%
\bibitem [{\citenamefont {Bennett}\ and\ \citenamefont {Wiesner}(1992)}]{bennett1992communication}%
  \BibitemOpen
  \bibfield  {author} {\bibinfo {author} {\bibfnamefont {C.~H.}\ \bibnamefont {Bennett}}\ and\ \bibinfo {author} {\bibfnamefont {S.~J.}\ \bibnamefont {Wiesner}},\ }\bibfield  {title} {\bibinfo {title} {Communication via one-and two-particle operators on einstein-podolsky-rosen states},\ }\href {https://doi.org/10.1103/PhysRevLett.69.2881} {\bibfield  {journal} {\bibinfo  {journal} {Physical review letters}\ }\textbf {\bibinfo {volume} {69}},\ \bibinfo {pages} {2881} (\bibinfo {year} {1992})}\BibitemShut {NoStop}%
\bibitem [{\citenamefont {Bru{\ss}}\ \emph {et~al.}(2004)\citenamefont {Bru{\ss}}, \citenamefont {D'Ariano}, \citenamefont {Lewenstein}, \citenamefont {Macchiavello}, \citenamefont {Sen},\ and\ \citenamefont {Sen}}]{bruss2004distributed}%
  \BibitemOpen
  \bibfield  {author} {\bibinfo {author} {\bibfnamefont {D.}~\bibnamefont {Bru{\ss}}}, \bibinfo {author} {\bibfnamefont {G.~M.}\ \bibnamefont {D'Ariano}}, \bibinfo {author} {\bibfnamefont {M.}~\bibnamefont {Lewenstein}}, \bibinfo {author} {\bibfnamefont {C.}~\bibnamefont {Macchiavello}}, \bibinfo {author} {\bibfnamefont {A.}~\bibnamefont {Sen}},\ and\ \bibinfo {author} {\bibfnamefont {U.}~\bibnamefont {Sen}},\ }\bibfield  {title} {\bibinfo {title} {Distributed quantum dense coding},\ }\href {https://doi.org/10.1103/PhysRevLett.93.210501} {\bibfield  {journal} {\bibinfo  {journal} {Physical review letters}\ }\textbf {\bibinfo {volume} {93}},\ \bibinfo {pages} {210501} (\bibinfo {year} {2004})}\BibitemShut {NoStop}%
\bibitem [{\citenamefont {Prabhu}\ \emph {et~al.}(2013)\citenamefont {Prabhu}, \citenamefont {Pati}, \citenamefont {Sen},\ and\ \citenamefont {Sen}}]{prabhu2013exclusion}%
  \BibitemOpen
  \bibfield  {author} {\bibinfo {author} {\bibfnamefont {R.}~\bibnamefont {Prabhu}}, \bibinfo {author} {\bibfnamefont {A.~K.}\ \bibnamefont {Pati}}, \bibinfo {author} {\bibfnamefont {A.}~\bibnamefont {Sen}},\ and\ \bibinfo {author} {\bibfnamefont {U.}~\bibnamefont {Sen}},\ }\bibfield  {title} {\bibinfo {title} {Exclusion principle for quantum dense coding},\ }\href {https://doi.org/10.1103/PhysRevA.87.052319} {\bibfield  {journal} {\bibinfo  {journal} {Physical Review A}\ }\textbf {\bibinfo {volume} {87}},\ \bibinfo {pages} {052319} (\bibinfo {year} {2013})}\BibitemShut {NoStop}%
\bibitem [{\citenamefont {Horodecki}\ \emph {et~al.}(2005)\citenamefont {Horodecki}, \citenamefont {Oppenheim},\ and\ \citenamefont {Winter}}]{horodecki2005partial}%
  \BibitemOpen
  \bibfield  {author} {\bibinfo {author} {\bibfnamefont {M.}~\bibnamefont {Horodecki}}, \bibinfo {author} {\bibfnamefont {J.}~\bibnamefont {Oppenheim}},\ and\ \bibinfo {author} {\bibfnamefont {A.}~\bibnamefont {Winter}},\ }\bibfield  {title} {\bibinfo {title} {Partial quantum information},\ }\href {https://doi.org/10.1038/nature03909} {\bibfield  {journal} {\bibinfo  {journal} {Nature}\ }\textbf {\bibinfo {volume} {436}},\ \bibinfo {pages} {673} (\bibinfo {year} {2005})}\BibitemShut {NoStop}%
\bibitem [{\citenamefont {Horodecki}\ \emph {et~al.}(2007)\citenamefont {Horodecki}, \citenamefont {Oppenheim},\ and\ \citenamefont {Winter}}]{horodecki2007quantum}%
  \BibitemOpen
  \bibfield  {author} {\bibinfo {author} {\bibfnamefont {M.}~\bibnamefont {Horodecki}}, \bibinfo {author} {\bibfnamefont {J.}~\bibnamefont {Oppenheim}},\ and\ \bibinfo {author} {\bibfnamefont {A.}~\bibnamefont {Winter}},\ }\bibfield  {title} {\bibinfo {title} {Quantum state merging and negative information},\ }\href {https://doi.org/10.1007/s00220-006-0118-x} {\bibfield  {journal} {\bibinfo  {journal} {Communications in Mathematical Physics}\ }\textbf {\bibinfo {volume} {269}},\ \bibinfo {pages} {107} (\bibinfo {year} {2007})}\BibitemShut {NoStop}%
\bibitem [{\citenamefont {Curty}\ \emph {et~al.}(2004)\citenamefont {Curty}, \citenamefont {Lewenstein},\ and\ \citenamefont {L{\"u}tkenhaus}}]{curty2004entanglement}%
  \BibitemOpen
  \bibfield  {author} {\bibinfo {author} {\bibfnamefont {M.}~\bibnamefont {Curty}}, \bibinfo {author} {\bibfnamefont {M.}~\bibnamefont {Lewenstein}},\ and\ \bibinfo {author} {\bibfnamefont {N.}~\bibnamefont {L{\"u}tkenhaus}},\ }\bibfield  {title} {\bibinfo {title} {Entanglement as a precondition for secure quantum key distribution},\ }\href {https://doi.org/10.1103/PhysRevLett.92.217903} {\bibfield  {journal} {\bibinfo  {journal} {Physical review letters}\ }\textbf {\bibinfo {volume} {92}},\ \bibinfo {pages} {217903} (\bibinfo {year} {2004})}\BibitemShut {NoStop}%
\bibitem [{\citenamefont {{\.Z}ukowski}\ \emph {et~al.}(1993)\citenamefont {{\.Z}ukowski}, \citenamefont {Zeilinger}, \citenamefont {Horne},\ and\ \citenamefont {Ekert}}]{zukowski1993event}%
  \BibitemOpen
  \bibfield  {author} {\bibinfo {author} {\bibfnamefont {M.}~\bibnamefont {{\.Z}ukowski}}, \bibinfo {author} {\bibfnamefont {A.}~\bibnamefont {Zeilinger}}, \bibinfo {author} {\bibfnamefont {M.~A.}\ \bibnamefont {Horne}},\ and\ \bibinfo {author} {\bibfnamefont {A.~K.}\ \bibnamefont {Ekert}},\ }\bibfield  {title} {\bibinfo {title} {‘‘event-ready-detectors’’bell experiment via entanglement swapping},\ }\href {https://doi.org/10.1103/PhysRevLett.71.4287} {\bibfield  {journal} {\bibinfo  {journal} {Physical Review Letters}\ }\textbf {\bibinfo {volume} {71}},\ \bibinfo {pages} {4287} (\bibinfo {year} {1993})}\BibitemShut {NoStop}%
\bibitem [{\citenamefont {Gisin}\ \emph {et~al.}(2002)\citenamefont {Gisin}, \citenamefont {Ribordy}, \citenamefont {Tittel},\ and\ \citenamefont {Zbinden}}]{gisin2002quantum}%
  \BibitemOpen
  \bibfield  {author} {\bibinfo {author} {\bibfnamefont {N.}~\bibnamefont {Gisin}}, \bibinfo {author} {\bibfnamefont {G.}~\bibnamefont {Ribordy}}, \bibinfo {author} {\bibfnamefont {W.}~\bibnamefont {Tittel}},\ and\ \bibinfo {author} {\bibfnamefont {H.}~\bibnamefont {Zbinden}},\ }\bibfield  {title} {\bibinfo {title} {Quantum cryptography},\ }\href {https://doi.org/10.1103/RevModPhys.74.145} {\bibfield  {journal} {\bibinfo  {journal} {Reviews of modern physics}\ }\textbf {\bibinfo {volume} {74}},\ \bibinfo {pages} {145} (\bibinfo {year} {2002})}\BibitemShut {NoStop}%
\bibitem [{\citenamefont {Bennett}\ \emph {et~al.}(2001)\citenamefont {Bennett}, \citenamefont {DiVincenzo}, \citenamefont {Shor}, \citenamefont {Smolin}, \citenamefont {Terhal},\ and\ \citenamefont {Wootters}}]{bennett2001remote}%
  \BibitemOpen
  \bibfield  {author} {\bibinfo {author} {\bibfnamefont {C.~H.}\ \bibnamefont {Bennett}}, \bibinfo {author} {\bibfnamefont {D.~P.}\ \bibnamefont {DiVincenzo}}, \bibinfo {author} {\bibfnamefont {P.~W.}\ \bibnamefont {Shor}}, \bibinfo {author} {\bibfnamefont {J.~A.}\ \bibnamefont {Smolin}}, \bibinfo {author} {\bibfnamefont {B.~M.}\ \bibnamefont {Terhal}},\ and\ \bibinfo {author} {\bibfnamefont {W.~K.}\ \bibnamefont {Wootters}},\ }\bibfield  {title} {\bibinfo {title} {Remote state preparation},\ }\href {https://doi.org/10.1103/PhysRevLett.87.077902} {\bibfield  {journal} {\bibinfo  {journal} {Physical Review Letters}\ }\textbf {\bibinfo {volume} {87}},\ \bibinfo {pages} {077902} (\bibinfo {year} {2001})}\BibitemShut {NoStop}%
\bibitem [{\citenamefont {Shi}\ and\ \citenamefont {Tomita}(2002)}]{shi2002remote}%
  \BibitemOpen
  \bibfield  {author} {\bibinfo {author} {\bibfnamefont {B.-S.}\ \bibnamefont {Shi}}\ and\ \bibinfo {author} {\bibfnamefont {A.}~\bibnamefont {Tomita}},\ }\bibfield  {title} {\bibinfo {title} {Remote state preparation of an entangled state},\ }\href {https://doi.org/10.1088/1464-4266/4/6/302} {\bibfield  {journal} {\bibinfo  {journal} {Journal of Optics B: Quantum and Semiclassical Optics}\ }\textbf {\bibinfo {volume} {4}},\ \bibinfo {pages} {380} (\bibinfo {year} {2002})}\BibitemShut {NoStop}%
\bibitem [{\citenamefont {Ye}\ \emph {et~al.}(2004)\citenamefont {Ye}, \citenamefont {Zhang},\ and\ \citenamefont {Guo}}]{ye2004faithful}%
  \BibitemOpen
  \bibfield  {author} {\bibinfo {author} {\bibfnamefont {M.-Y.}\ \bibnamefont {Ye}}, \bibinfo {author} {\bibfnamefont {Y.-S.}\ \bibnamefont {Zhang}},\ and\ \bibinfo {author} {\bibfnamefont {G.-C.}\ \bibnamefont {Guo}},\ }\bibfield  {title} {\bibinfo {title} {Faithful remote state preparation using finite classical bits and a nonmaximally entangled state},\ }\href {https://doi.org/10.1103/PhysRevA.69.022310} {\bibfield  {journal} {\bibinfo  {journal} {Physical Review A}\ }\textbf {\bibinfo {volume} {69}},\ \bibinfo {pages} {022310} (\bibinfo {year} {2004})}\BibitemShut {NoStop}%
\bibitem [{\citenamefont {Rui-Juan}\ \emph {et~al.}(2010)\citenamefont {Rui-Juan}, \citenamefont {Ming}, \citenamefont {Shao-Ming},\ and\ \citenamefont {Xian-Qing}}]{rui2010estimation}%
  \BibitemOpen
  \bibfield  {author} {\bibinfo {author} {\bibfnamefont {G.}~\bibnamefont {Rui-Juan}}, \bibinfo {author} {\bibfnamefont {L.}~\bibnamefont {Ming}}, \bibinfo {author} {\bibfnamefont {F.}~\bibnamefont {Shao-Ming}},\ and\ \bibinfo {author} {\bibfnamefont {L.-j.}\ \bibnamefont {Xian-Qing}},\ }\bibfield  {title} {\bibinfo {title} {On estimation of fully entangled fraction},\ }\href {https://doi.org/10.1088/0253-6102/53/2/12} {\bibfield  {journal} {\bibinfo  {journal} {Communications in Theoretical Physics}\ }\textbf {\bibinfo {volume} {53}},\ \bibinfo {pages} {265} (\bibinfo {year} {2010})}\BibitemShut {NoStop}%
\bibitem [{\citenamefont {Li}\ \emph {et~al.}(2008)\citenamefont {Li}, \citenamefont {Fei},\ and\ \citenamefont {Wang}}]{li2008upper}%
  \BibitemOpen
  \bibfield  {author} {\bibinfo {author} {\bibfnamefont {M.}~\bibnamefont {Li}}, \bibinfo {author} {\bibfnamefont {S.-M.}\ \bibnamefont {Fei}},\ and\ \bibinfo {author} {\bibfnamefont {Z.-X.}\ \bibnamefont {Wang}},\ }\bibfield  {title} {\bibinfo {title} {Upper bound of the fully entangled fraction},\ }\href {https://doi.org/10.1103/PhysRevA.78.032332} {\bibfield  {journal} {\bibinfo  {journal} {Physical Review A}\ }\textbf {\bibinfo {volume} {78}},\ \bibinfo {pages} {032332} (\bibinfo {year} {2008})}\BibitemShut {NoStop}%
\bibitem [{\citenamefont {Huang}\ \emph {et~al.}(2016)\citenamefont {Huang}, \citenamefont {Jing},\ and\ \citenamefont {Zhang}}]{huang2016upper}%
  \BibitemOpen
  \bibfield  {author} {\bibinfo {author} {\bibfnamefont {X.~F.}\ \bibnamefont {Huang}}, \bibinfo {author} {\bibfnamefont {N.~H.}\ \bibnamefont {Jing}},\ and\ \bibinfo {author} {\bibfnamefont {T.~G.}\ \bibnamefont {Zhang}},\ }\bibfield  {title} {\bibinfo {title} {An upper bound of fully entangled fraction of mixed states},\ }\href {https://doi.org/10.1088/0253-6102/65/6/701} {\bibfield  {journal} {\bibinfo  {journal} {Communications in Theoretical Physics}\ }\textbf {\bibinfo {volume} {65}},\ \bibinfo {pages} {701} (\bibinfo {year} {2016})}\BibitemShut {NoStop}%
\bibitem [{\citenamefont {Zhao}(2015)}]{zhao2015maximally}%
  \BibitemOpen
  \bibfield  {author} {\bibinfo {author} {\bibfnamefont {M.~J.}\ \bibnamefont {Zhao}},\ }\bibfield  {title} {\bibinfo {title} {Maximally entangled states and fully entangled fraction},\ }\href {https://doi.org/10.1103/PhysRevA.91.012310} {\bibfield  {journal} {\bibinfo  {journal} {Physical Review A}\ }\textbf {\bibinfo {volume} {91}},\ \bibinfo {pages} {012310} (\bibinfo {year} {2015})}\BibitemShut {NoStop}%
\bibitem [{\citenamefont {Zhao}\ \emph {et~al.}(2010)\citenamefont {Zhao}, \citenamefont {Li}, \citenamefont {Fei},\ and\ \citenamefont {Wang}}]{Zhao_2010}%
  \BibitemOpen
  \bibfield  {author} {\bibinfo {author} {\bibfnamefont {M.~J.}\ \bibnamefont {Zhao}}, \bibinfo {author} {\bibfnamefont {Z.~G.}\ \bibnamefont {Li}}, \bibinfo {author} {\bibfnamefont {S.~M.}\ \bibnamefont {Fei}},\ and\ \bibinfo {author} {\bibfnamefont {Z.~X.}\ \bibnamefont {Wang}},\ }\bibfield  {title} {\bibinfo {title} {A note on fully entangled fraction},\ }\href {https://doi.org/10.1088/1751-8113/43/27/275203} {\bibfield  {journal} {\bibinfo  {journal} {Journal of Physics A: Mathematical and Theoretical}\ }\textbf {\bibinfo {volume} {43}},\ \bibinfo {pages} {275203} (\bibinfo {year} {2010})}\BibitemShut {NoStop}%
\bibitem [{\citenamefont {Grondalski}\ \emph {et~al.}(2002)\citenamefont {Grondalski}, \citenamefont {Etlinger},\ and\ \citenamefont {James}}]{grondalski2002fully}%
  \BibitemOpen
  \bibfield  {author} {\bibinfo {author} {\bibfnamefont {J.}~\bibnamefont {Grondalski}}, \bibinfo {author} {\bibfnamefont {D.}~\bibnamefont {Etlinger}},\ and\ \bibinfo {author} {\bibfnamefont {D.}~\bibnamefont {James}},\ }\bibfield  {title} {\bibinfo {title} {The fully entangled fraction as an inclusive measure of entanglement applications},\ }\href {https://doi.org/10.1016/S0375-9601(02)00884-8} {\bibfield  {journal} {\bibinfo  {journal} {Physics Letters A}\ }\textbf {\bibinfo {volume} {300}},\ \bibinfo {pages} {573} (\bibinfo {year} {2002})}\BibitemShut {NoStop}%
\bibitem [{\citenamefont {Patro}\ \emph {et~al.}(2022)\citenamefont {Patro}, \citenamefont {Mukherjee}, \citenamefont {Siddiqui}, \citenamefont {Chakrabarty},\ and\ \citenamefont {Ganguly}}]{patro2022absolute}%
  \BibitemOpen
  \bibfield  {author} {\bibinfo {author} {\bibfnamefont {T.}~\bibnamefont {Patro}}, \bibinfo {author} {\bibfnamefont {K.}~\bibnamefont {Mukherjee}}, \bibinfo {author} {\bibfnamefont {M.~A.}\ \bibnamefont {Siddiqui}}, \bibinfo {author} {\bibfnamefont {I.}~\bibnamefont {Chakrabarty}},\ and\ \bibinfo {author} {\bibfnamefont {N.}~\bibnamefont {Ganguly}},\ }\bibfield  {title} {\bibinfo {title} {Absolute fully entangled fraction from spectrum},\ }\href {https://doi.org/10.1140/epjd/s10053-022-00458-8} {\bibfield  {journal} {\bibinfo  {journal} {The European Physical Journal D}\ }\textbf {\bibinfo {volume} {76}},\ \bibinfo {pages} {127} (\bibinfo {year} {2022})}\BibitemShut {NoStop}%
\bibitem [{\citenamefont {Ganguly}\ \emph {et~al.}(2011)\citenamefont {Ganguly}, \citenamefont {Adhikari}, \citenamefont {Majumdar},\ and\ \citenamefont {Chatterjee}}]{ganguly2011entanglement}%
  \BibitemOpen
  \bibfield  {author} {\bibinfo {author} {\bibfnamefont {N.}~\bibnamefont {Ganguly}}, \bibinfo {author} {\bibfnamefont {S.}~\bibnamefont {Adhikari}}, \bibinfo {author} {\bibfnamefont {A.}~\bibnamefont {Majumdar}},\ and\ \bibinfo {author} {\bibfnamefont {J.}~\bibnamefont {Chatterjee}},\ }\bibfield  {title} {\bibinfo {title} {Entanglement witness operator for quantum teleportation},\ }\href {https://doi.org/10.1103/PhysRevLett.107.270501} {\bibfield  {journal} {\bibinfo  {journal} {Physical Review Letters}\ }\textbf {\bibinfo {volume} {107}},\ \bibinfo {pages} {270501} (\bibinfo {year} {2011})}\BibitemShut {NoStop}%
\bibitem [{\citenamefont {Cavalcanti}\ \emph {et~al.}(2013)\citenamefont {Cavalcanti}, \citenamefont {Acin}, \citenamefont {Brunner},\ and\ \citenamefont {V{\'e}rtesi}}]{cavalcanti2013all}%
  \BibitemOpen
  \bibfield  {author} {\bibinfo {author} {\bibfnamefont {D.}~\bibnamefont {Cavalcanti}}, \bibinfo {author} {\bibfnamefont {A.}~\bibnamefont {Acin}}, \bibinfo {author} {\bibfnamefont {N.}~\bibnamefont {Brunner}},\ and\ \bibinfo {author} {\bibfnamefont {T.}~\bibnamefont {V{\'e}rtesi}},\ }\bibfield  {title} {\bibinfo {title} {All quantum states useful for teleportation are nonlocal resources},\ }\href {https://doi.org/10.1103/PhysRevA.87.042104} {\bibfield  {journal} {\bibinfo  {journal} {Physical Review A}\ }\textbf {\bibinfo {volume} {87}},\ \bibinfo {pages} {042104} (\bibinfo {year} {2013})}\BibitemShut {NoStop}%
\bibitem [{\citenamefont {Ghosal}\ \emph {et~al.}(2025)\citenamefont {Ghosal}, \citenamefont {Ghai}, \citenamefont {Saha}, \citenamefont {Ghosh},\ and\ \citenamefont {Alimuddin}}]{ghosal2025repeater}%
  \BibitemOpen
  \bibfield  {author} {\bibinfo {author} {\bibfnamefont {A.}~\bibnamefont {Ghosal}}, \bibinfo {author} {\bibfnamefont {J.}~\bibnamefont {Ghai}}, \bibinfo {author} {\bibfnamefont {T.}~\bibnamefont {Saha}}, \bibinfo {author} {\bibfnamefont {S.}~\bibnamefont {Ghosh}},\ and\ \bibinfo {author} {\bibfnamefont {M.}~\bibnamefont {Alimuddin}},\ }\bibfield  {title} {\bibinfo {title} {Repeater-based quantum communication protocol: Maximizing teleportation fidelity with minimal entanglement},\ }\href {https://doi.org/10.1103/PhysRevLett.134.160803} {\bibfield  {journal} {\bibinfo  {journal} {Physical Review Letters}\ }\textbf {\bibinfo {volume} {134}},\ \bibinfo {pages} {160803} (\bibinfo {year} {2025})}\BibitemShut {NoStop}%
\bibitem [{\citenamefont {Pal}\ and\ \citenamefont {Ghosh}(2015)}]{pal2015non}%
  \BibitemOpen
  \bibfield  {author} {\bibinfo {author} {\bibfnamefont {R.}~\bibnamefont {Pal}}\ and\ \bibinfo {author} {\bibfnamefont {S.}~\bibnamefont {Ghosh}},\ }\bibfield  {title} {\bibinfo {title} {Non-locality breaking qubit channels: the case for chsh inequality},\ }\href {https://doi.org/10.1088/1751-8113/48/15/155302} {\bibfield  {journal} {\bibinfo  {journal} {Journal of Physics A: Mathematical and Theoretical}\ }\textbf {\bibinfo {volume} {48}},\ \bibinfo {pages} {155302} (\bibinfo {year} {2015})}\BibitemShut {NoStop}%
\bibitem [{\citenamefont {Luo}\ \emph {et~al.}(2022)\citenamefont {Luo}, \citenamefont {Li},\ and\ \citenamefont {Xi}}]{luo2022coherence}%
  \BibitemOpen
  \bibfield  {author} {\bibinfo {author} {\bibfnamefont {Y.}~\bibnamefont {Luo}}, \bibinfo {author} {\bibfnamefont {Y.}~\bibnamefont {Li}},\ and\ \bibinfo {author} {\bibfnamefont {Z.}~\bibnamefont {Xi}},\ }\bibfield  {title} {\bibinfo {title} {Coherence-breaking superchannels},\ }\href {https://doi.org/10.1007/s11128-022-03511-y} {\bibfield  {journal} {\bibinfo  {journal} {Quantum Information Processing}\ }\textbf {\bibinfo {volume} {21}},\ \bibinfo {pages} {176} (\bibinfo {year} {2022})}\BibitemShut {NoStop}%
\bibitem [{\citenamefont {Heinosaari}\ \emph {et~al.}(2015)\citenamefont {Heinosaari}, \citenamefont {Kiukas}, \citenamefont {Reitzner},\ and\ \citenamefont {Schultz}}]{heinosaari2015incompatibility}%
  \BibitemOpen
  \bibfield  {author} {\bibinfo {author} {\bibfnamefont {T.}~\bibnamefont {Heinosaari}}, \bibinfo {author} {\bibfnamefont {J.}~\bibnamefont {Kiukas}}, \bibinfo {author} {\bibfnamefont {D.}~\bibnamefont {Reitzner}},\ and\ \bibinfo {author} {\bibfnamefont {J.}~\bibnamefont {Schultz}},\ }\bibfield  {title} {\bibinfo {title} {Incompatibility breaking quantum channels},\ }\href {https://doi.org/10.1088/1751-8113/48/43/435301} {\bibfield  {journal} {\bibinfo  {journal} {Journal of Physics A: Mathematical and Theoretical}\ }\textbf {\bibinfo {volume} {48}},\ \bibinfo {pages} {435301} (\bibinfo {year} {2015})}\BibitemShut {NoStop}%
\bibitem [{\citenamefont {Ku}\ \emph {et~al.}(2022)\citenamefont {Ku}, \citenamefont {Kadlec}, \citenamefont {{\v{C}}ernoch}, \citenamefont {Quintino}, \citenamefont {Zhou}, \citenamefont {Lemr}, \citenamefont {Lambert}, \citenamefont {Miranowicz}, \citenamefont {Chen}, \citenamefont {Nori} \emph {et~al.}}]{ku2022quantifying}%
  \BibitemOpen
  \bibfield  {author} {\bibinfo {author} {\bibfnamefont {H.-Y.}\ \bibnamefont {Ku}}, \bibinfo {author} {\bibfnamefont {J.}~\bibnamefont {Kadlec}}, \bibinfo {author} {\bibfnamefont {A.}~\bibnamefont {{\v{C}}ernoch}}, \bibinfo {author} {\bibfnamefont {M.~T.}\ \bibnamefont {Quintino}}, \bibinfo {author} {\bibfnamefont {W.}~\bibnamefont {Zhou}}, \bibinfo {author} {\bibfnamefont {K.}~\bibnamefont {Lemr}}, \bibinfo {author} {\bibfnamefont {N.}~\bibnamefont {Lambert}}, \bibinfo {author} {\bibfnamefont {A.}~\bibnamefont {Miranowicz}}, \bibinfo {author} {\bibfnamefont {S.-L.}\ \bibnamefont {Chen}}, \bibinfo {author} {\bibfnamefont {F.}~\bibnamefont {Nori}}, \emph {et~al.},\ }\bibfield  {title} {\bibinfo {title} {Quantifying quantumness of channels without entanglement},\ }\href {https://doi.org/10.1103/PRXQuantum.3.020338} {\bibfield  {journal} {\bibinfo  {journal} {PRX Quantum}\ }\textbf {\bibinfo {volume} {3}},\ \bibinfo {pages} {020338} (\bibinfo {year} {2022})}\BibitemShut {NoStop}%
\bibitem [{\citenamefont {Horodecki}\ \emph {et~al.}(2003)\citenamefont {Horodecki}, \citenamefont {Shor},\ and\ \citenamefont {Ruskai}}]{horodecki2003entanglement}%
  \BibitemOpen
  \bibfield  {author} {\bibinfo {author} {\bibfnamefont {M.}~\bibnamefont {Horodecki}}, \bibinfo {author} {\bibfnamefont {P.~W.}\ \bibnamefont {Shor}},\ and\ \bibinfo {author} {\bibfnamefont {M.~B.}\ \bibnamefont {Ruskai}},\ }\bibfield  {title} {\bibinfo {title} {Entanglement breaking channels},\ }\href {https://doi.org/10.1142/S0129055X03001709} {\bibfield  {journal} {\bibinfo  {journal} {Reviews in Mathematical Physics}\ }\textbf {\bibinfo {volume} {15}},\ \bibinfo {pages} {629} (\bibinfo {year} {2003})}\BibitemShut {NoStop}%
\bibitem [{\citenamefont {Morav{\v{c}}{\'\i}kov{\'a}}\ and\ \citenamefont {Ziman}(2010)}]{moravvcikova2010entanglement}%
  \BibitemOpen
  \bibfield  {author} {\bibinfo {author} {\bibfnamefont {L.}~\bibnamefont {Morav{\v{c}}{\'\i}kov{\'a}}}\ and\ \bibinfo {author} {\bibfnamefont {M.}~\bibnamefont {Ziman}},\ }\bibfield  {title} {\bibinfo {title} {Entanglement-annihilating and entanglement-breaking channels},\ }\href {https://doi.org/10.1088/1751-8113/43/27/275306} {\bibfield  {journal} {\bibinfo  {journal} {Journal of Physics A: Mathematical and Theoretical}\ }\textbf {\bibinfo {volume} {43}},\ \bibinfo {pages} {275306} (\bibinfo {year} {2010})}\BibitemShut {NoStop}%
\bibitem [{\citenamefont {Chru{\'s}ci{\'n}ski}\ and\ \citenamefont {Kossakowski}(2006)}]{chruscinski2006partially}%
  \BibitemOpen
  \bibfield  {author} {\bibinfo {author} {\bibfnamefont {D.}~\bibnamefont {Chru{\'s}ci{\'n}ski}}\ and\ \bibinfo {author} {\bibfnamefont {A.}~\bibnamefont {Kossakowski}},\ }\bibfield  {title} {\bibinfo {title} {On partially entanglement breaking channels},\ }\href {https://doi.org/10.1007/s11080-006-7264-7} {\bibfield  {journal} {\bibinfo  {journal} {Open Systems \& Information Dynamics}\ }\textbf {\bibinfo {volume} {13}},\ \bibinfo {pages} {17} (\bibinfo {year} {2006})}\BibitemShut {NoStop}%
\bibitem [{\citenamefont {Patro}\ \emph {et~al.}(2024)\citenamefont {Patro}, \citenamefont {Mukherjee},\ and\ \citenamefont {Ganguly}}]{patro2024quantum}%
  \BibitemOpen
  \bibfield  {author} {\bibinfo {author} {\bibfnamefont {T.}~\bibnamefont {Patro}}, \bibinfo {author} {\bibfnamefont {K.}~\bibnamefont {Mukherjee}},\ and\ \bibinfo {author} {\bibfnamefont {N.}~\bibnamefont {Ganguly}},\ }\bibfield  {title} {\bibinfo {title} {Quantum channels and some absolute properties of quantum states},\ }\href {https://doi.org/10.1007/s11128-024-04439-1} {\bibfield  {journal} {\bibinfo  {journal} {Quantum Information Processing}\ }\textbf {\bibinfo {volume} {23}},\ \bibinfo {pages} {1} (\bibinfo {year} {2024})}\BibitemShut {NoStop}%
\bibitem [{\citenamefont {Mallick}\ \emph {et~al.}(2024)\citenamefont {Mallick}, \citenamefont {Ganguly},\ and\ \citenamefont {Majumdar}}]{mallick2024characterization}%
  \BibitemOpen
  \bibfield  {author} {\bibinfo {author} {\bibfnamefont {B.}~\bibnamefont {Mallick}}, \bibinfo {author} {\bibfnamefont {N.}~\bibnamefont {Ganguly}},\ and\ \bibinfo {author} {\bibfnamefont {A.~S.}\ \bibnamefont {Majumdar}},\ }\bibfield  {title} {\bibinfo {title} {On the characterization of schmidt number breaking and annihilating channels},\ }\bibfield  {journal} {\bibinfo  {journal} {arXiv preprint arXiv:2411.19315}\ }\href {https://doi.org/10.48550/arXiv.2411.19315} {10.48550/arXiv.2411.19315} (\bibinfo {year} {2024})\BibitemShut {NoStop}%
\bibitem [{\citenamefont {Srinidhi}\ \emph {et~al.}(2024)\citenamefont {Srinidhi}, \citenamefont {Chakrabarty}, \citenamefont {Bhattacharya},\ and\ \citenamefont {Ganguly}}]{srinidhi2024quantum}%
  \BibitemOpen
  \bibfield  {author} {\bibinfo {author} {\bibfnamefont {P.}~\bibnamefont {Srinidhi}}, \bibinfo {author} {\bibfnamefont {I.}~\bibnamefont {Chakrabarty}}, \bibinfo {author} {\bibfnamefont {S.}~\bibnamefont {Bhattacharya}},\ and\ \bibinfo {author} {\bibfnamefont {N.}~\bibnamefont {Ganguly}},\ }\bibfield  {title} {\bibinfo {title} {Quantum channels that destroy negative conditional entropy},\ }\href {https://doi.org/10.1103/PhysRevA.110.042423} {\bibfield  {journal} {\bibinfo  {journal} {Physical Review A}\ }\textbf {\bibinfo {volume} {110}},\ \bibinfo {pages} {042423} (\bibinfo {year} {2024})}\BibitemShut {NoStop}%
\bibitem [{\citenamefont {Muhuri}\ \emph {et~al.}(2023)\citenamefont {Muhuri}, \citenamefont {Patra}, \citenamefont {Gupta},\ and\ \citenamefont {De}}]{muhuri2023information}%
  \BibitemOpen
  \bibfield  {author} {\bibinfo {author} {\bibfnamefont {A.}~\bibnamefont {Muhuri}}, \bibinfo {author} {\bibfnamefont {A.}~\bibnamefont {Patra}}, \bibinfo {author} {\bibfnamefont {R.}~\bibnamefont {Gupta}},\ and\ \bibinfo {author} {\bibfnamefont {A.~S.}\ \bibnamefont {De}},\ }\bibfield  {title} {\bibinfo {title} {Information theoretic resource-breaking channels},\ }\bibfield  {journal} {\bibinfo  {journal} {arXiv preprint arXiv:2309.03108}\ }\href {https://doi.org/10.48550/arXiv.2309.03108} {10.48550/arXiv.2309.03108} (\bibinfo {year} {2023})\BibitemShut {NoStop}%
\bibitem [{\citenamefont {Friis}\ \emph {et~al.}(2017)\citenamefont {Friis}, \citenamefont {Bulusu},\ and\ \citenamefont {Bertlmann}}]{friis2017geometry}%
  \BibitemOpen
  \bibfield  {author} {\bibinfo {author} {\bibfnamefont {N.}~\bibnamefont {Friis}}, \bibinfo {author} {\bibfnamefont {S.}~\bibnamefont {Bulusu}},\ and\ \bibinfo {author} {\bibfnamefont {R.~A.}\ \bibnamefont {Bertlmann}},\ }\bibfield  {title} {\bibinfo {title} {Geometry of two-qubit states with negative conditional entropy},\ }\href {https://doi.org/10.1088/1751-8121/aa5dfd} {\bibfield  {journal} {\bibinfo  {journal} {Journal of Physics A: Mathematical and Theoretical}\ }\textbf {\bibinfo {volume} {50}},\ \bibinfo {pages} {125301} (\bibinfo {year} {2017})}\BibitemShut {NoStop}%
\bibitem [{\citenamefont {Kumar}\ and\ \citenamefont {Ganguly}(2023)}]{kumar2023quantum}%
  \BibitemOpen
  \bibfield  {author} {\bibinfo {author} {\bibfnamefont {K.}~\bibnamefont {Kumar}}\ and\ \bibinfo {author} {\bibfnamefont {N.}~\bibnamefont {Ganguly}},\ }\bibfield  {title} {\bibinfo {title} {Quantum conditional entropies and steerability of states with maximally mixed marginals},\ }\href {https://doi.org/10.1103/PhysRevA.107.032206} {\bibfield  {journal} {\bibinfo  {journal} {Physical Review A}\ }\textbf {\bibinfo {volume} {107}},\ \bibinfo {pages} {032206} (\bibinfo {year} {2023})}\BibitemShut {NoStop}%
\bibitem [{\citenamefont {Clauser}\ \emph {et~al.}(1969)\citenamefont {Clauser}, \citenamefont {Horne}, \citenamefont {Shimony},\ and\ \citenamefont {Holt}}]{clauser1969proposed}%
  \BibitemOpen
  \bibfield  {author} {\bibinfo {author} {\bibfnamefont {J.~F.}\ \bibnamefont {Clauser}}, \bibinfo {author} {\bibfnamefont {M.~A.}\ \bibnamefont {Horne}}, \bibinfo {author} {\bibfnamefont {A.}~\bibnamefont {Shimony}},\ and\ \bibinfo {author} {\bibfnamefont {R.~A.}\ \bibnamefont {Holt}},\ }\bibfield  {title} {\bibinfo {title} {Proposed experiment to test local hidden-variable theories},\ }\href {https://doi.org/10.1103/PhysRevLett.23.880} {\bibfield  {journal} {\bibinfo  {journal} {Physical review letters}\ }\textbf {\bibinfo {volume} {23}},\ \bibinfo {pages} {880} (\bibinfo {year} {1969})}\BibitemShut {NoStop}%
\bibitem [{\citenamefont {Cavalcanti}\ \emph {et~al.}(2009)\citenamefont {Cavalcanti}, \citenamefont {Jones}, \citenamefont {Wiseman},\ and\ \citenamefont {Reid}}]{cavalcanti2009experimental}%
  \BibitemOpen
  \bibfield  {author} {\bibinfo {author} {\bibfnamefont {E.~G.}\ \bibnamefont {Cavalcanti}}, \bibinfo {author} {\bibfnamefont {S.~J.}\ \bibnamefont {Jones}}, \bibinfo {author} {\bibfnamefont {H.~M.}\ \bibnamefont {Wiseman}},\ and\ \bibinfo {author} {\bibfnamefont {M.~D.}\ \bibnamefont {Reid}},\ }\bibfield  {title} {\bibinfo {title} {Experimental criteria for steering and the einstein-podolsky-rosen paradox},\ }\href {https://doi.org/10.1103/PhysRevA.80.032112} {\bibfield  {journal} {\bibinfo  {journal} {Physical Review A}\ }\textbf {\bibinfo {volume} {80}},\ \bibinfo {pages} {032112} (\bibinfo {year} {2009})}\BibitemShut {NoStop}%
\bibitem [{\citenamefont {Kumar}\ \emph {et~al.}(2025)\citenamefont {Kumar}, \citenamefont {Chakrabarty},\ and\ \citenamefont {Ganguly}}]{kumar2024quantum}%
  \BibitemOpen
  \bibfield  {author} {\bibinfo {author} {\bibfnamefont {K.}~\bibnamefont {Kumar}}, \bibinfo {author} {\bibfnamefont {I.}~\bibnamefont {Chakrabarty}},\ and\ \bibinfo {author} {\bibfnamefont {N.}~\bibnamefont {Ganguly}},\ }\bibfield  {title} {\bibinfo {title} {On fully entangled fraction and quantum conditional entropies for states with maximally mixed marginals},\ }\href {https://doi.org/https://doi.org/10.1007/s11128-025-04695-9} {\bibfield  {journal} {\bibinfo  {journal} {Quantum Information Processing}\ }\textbf {\bibinfo {volume} {24}},\ \bibinfo {pages} {79} (\bibinfo {year} {2025})}\BibitemShut {NoStop}%
\bibitem [{\citenamefont {Vollbrecht}\ and\ \citenamefont {Wolf}(2002)}]{vollbrecht2002conditional}%
  \BibitemOpen
  \bibfield  {author} {\bibinfo {author} {\bibfnamefont {K.~G.~H.}\ \bibnamefont {Vollbrecht}}\ and\ \bibinfo {author} {\bibfnamefont {M.~M.}\ \bibnamefont {Wolf}},\ }\bibfield  {title} {\bibinfo {title} {Conditional entropies and their relation to entanglement criteria},\ }\href {https://doi.org/10.1063/1.1498490} {\bibfield  {journal} {\bibinfo  {journal} {Journal of Mathematical Physics}\ }\textbf {\bibinfo {volume} {43}},\ \bibinfo {pages} {4299} (\bibinfo {year} {2002})}\BibitemShut {NoStop}%
\bibitem [{\citenamefont {Horodecki}\ \emph {et~al.}(1997)\citenamefont {Horodecki}, \citenamefont {Horodecki},\ and\ \citenamefont {Horodecki}}]{horodecki1997inseparable}%
  \BibitemOpen
  \bibfield  {author} {\bibinfo {author} {\bibfnamefont {M.}~\bibnamefont {Horodecki}}, \bibinfo {author} {\bibfnamefont {P.}~\bibnamefont {Horodecki}},\ and\ \bibinfo {author} {\bibfnamefont {R.}~\bibnamefont {Horodecki}},\ }\bibfield  {title} {\bibinfo {title} {Inseparable two spin-$\frac{1}{2}$ density matrices can be distilled to a singlet form},\ }\href {https://doi.org/10.1103/PhysRevLett.78.574} {\bibfield  {journal} {\bibinfo  {journal} {Physical Review Letters}\ }\textbf {\bibinfo {volume} {78}},\ \bibinfo {pages} {574} (\bibinfo {year} {1997})}\BibitemShut {NoStop}%
\bibitem [{\citenamefont {Vidal}\ \emph {et~al.}(2000)\citenamefont {Vidal}, \citenamefont {Jonathan},\ and\ \citenamefont {Nielsen}}]{vidal2000approximate}%
  \BibitemOpen
  \bibfield  {author} {\bibinfo {author} {\bibfnamefont {G.}~\bibnamefont {Vidal}}, \bibinfo {author} {\bibfnamefont {D.}~\bibnamefont {Jonathan}},\ and\ \bibinfo {author} {\bibfnamefont {M.}~\bibnamefont {Nielsen}},\ }\bibfield  {title} {\bibinfo {title} {Approximate transformations and robust manipulation of bipartite pure-state entanglement},\ }\href@noop {} {\bibfield  {journal} {\bibinfo  {journal} {Physical Review A}\ }\textbf {\bibinfo {volume} {62}},\ \bibinfo {pages} {012304} (\bibinfo {year} {2000})}\BibitemShut {NoStop}%
\bibitem [{\citenamefont {Holmes}(2012)}]{holmes2012geometric}%
  \BibitemOpen
  \bibfield  {author} {\bibinfo {author} {\bibfnamefont {R.~B.}\ \bibnamefont {Holmes}},\ }\href@noop {} {\emph {\bibinfo {title} {Geometric functional analysis and its applications}}},\ Vol.~\bibinfo {volume} {24}\ (\bibinfo  {publisher} {Springer Science \& Business Media},\ \bibinfo {year} {2012})\BibitemShut {NoStop}%
\bibitem [{\citenamefont {Rudin}(1976)}]{rudin1976principles}%
  \BibitemOpen
  \bibfield  {author} {\bibinfo {author} {\bibfnamefont {W.}~\bibnamefont {Rudin}},\ }\href@noop {} {\emph {\bibinfo {title} {Principles of Mathematical Analysis}}}\ (\bibinfo  {publisher} {McGrawHill, New York},\ \bibinfo {year} {1976})\BibitemShut {NoStop}%
\bibitem [{\citenamefont {Watrous}(2018)}]{watrous2018theory}%
  \BibitemOpen
  \bibfield  {author} {\bibinfo {author} {\bibfnamefont {J.}~\bibnamefont {Watrous}},\ }\href@noop {} {\emph {\bibinfo {title} {The theory of quantum information}}}\ (\bibinfo  {publisher} {Cambridge university press},\ \bibinfo {year} {2018})\BibitemShut {NoStop}%
\bibitem [{\citenamefont {Quan}\ \emph {et~al.}(2016)\citenamefont {Quan}, \citenamefont {Zhu}, \citenamefont {Liu}, \citenamefont {Fei}, \citenamefont {Fan},\ and\ \citenamefont {Yang}}]{quan2016steering}%
  \BibitemOpen
  \bibfield  {author} {\bibinfo {author} {\bibfnamefont {Q.}~\bibnamefont {Quan}}, \bibinfo {author} {\bibfnamefont {H.}~\bibnamefont {Zhu}}, \bibinfo {author} {\bibfnamefont {S.-Y.}\ \bibnamefont {Liu}}, \bibinfo {author} {\bibfnamefont {S.-M.}\ \bibnamefont {Fei}}, \bibinfo {author} {\bibfnamefont {H.}~\bibnamefont {Fan}},\ and\ \bibinfo {author} {\bibfnamefont {W.-L.}\ \bibnamefont {Yang}},\ }\bibfield  {title} {\bibinfo {title} {Steering bell-diagonal states},\ }\href {https://doi.org/10.1038/srep22025} {\bibfield  {journal} {\bibinfo  {journal} {Scientific reports}\ }\textbf {\bibinfo {volume} {6}},\ \bibinfo {pages} {22025} (\bibinfo {year} {2016})}\BibitemShut {NoStop}%
\bibitem [{\citenamefont {Riccardi}\ \emph {et~al.}(2021)\citenamefont {Riccardi}, \citenamefont {Jones}, \citenamefont {Yu}, \citenamefont {G{\"u}hne},\ and\ \citenamefont {Kirby}}]{riccardi2021exploring}%
  \BibitemOpen
  \bibfield  {author} {\bibinfo {author} {\bibfnamefont {G.}~\bibnamefont {Riccardi}}, \bibinfo {author} {\bibfnamefont {D.~E.}\ \bibnamefont {Jones}}, \bibinfo {author} {\bibfnamefont {X.-D.}\ \bibnamefont {Yu}}, \bibinfo {author} {\bibfnamefont {O.}~\bibnamefont {G{\"u}hne}},\ and\ \bibinfo {author} {\bibfnamefont {B.~T.}\ \bibnamefont {Kirby}},\ }\bibfield  {title} {\bibinfo {title} {Exploring the relationship between the faithfulness and entanglement of two qubits},\ }\href {https://doi.org/10.1103/PhysRevA.103.042417} {\bibfield  {journal} {\bibinfo  {journal} {Physical Review A}\ }\textbf {\bibinfo {volume} {103}},\ \bibinfo {pages} {042417} (\bibinfo {year} {2021})}\BibitemShut {NoStop}%
\bibitem [{\citenamefont {G{\"u}hne}\ \emph {et~al.}(2021)\citenamefont {G{\"u}hne}, \citenamefont {Mao},\ and\ \citenamefont {Yu}}]{guhne2021geometry}%
  \BibitemOpen
  \bibfield  {author} {\bibinfo {author} {\bibfnamefont {O.}~\bibnamefont {G{\"u}hne}}, \bibinfo {author} {\bibfnamefont {Y.}~\bibnamefont {Mao}},\ and\ \bibinfo {author} {\bibfnamefont {X.-D.}\ \bibnamefont {Yu}},\ }\bibfield  {title} {\bibinfo {title} {Geometry of faithful entanglement},\ }\href {https://doi.org/10.1103/PhysRevLett.126.140503} {\bibfield  {journal} {\bibinfo  {journal} {Physical Review Letters}\ }\textbf {\bibinfo {volume} {126}},\ \bibinfo {pages} {140503} (\bibinfo {year} {2021})}\BibitemShut {NoStop}%
\bibitem [{\citenamefont {Ruskai}(2002)}]{ruskai2002inequalities}%
  \BibitemOpen
  \bibfield  {author} {\bibinfo {author} {\bibfnamefont {M.~B.}\ \bibnamefont {Ruskai}},\ }\bibfield  {title} {\bibinfo {title} {Inequalities for quantum entropy: A review with conditions for equality},\ }\href {https://doi.org/10.1063/1.1497701} {\bibfield  {journal} {\bibinfo  {journal} {Journal of Mathematical Physics}\ }\textbf {\bibinfo {volume} {43}},\ \bibinfo {pages} {4358} (\bibinfo {year} {2002})}\BibitemShut {NoStop}%
\bibitem [{\citenamefont {Vedral}\ \emph {et~al.}(1997)\citenamefont {Vedral}, \citenamefont {Plenio}, \citenamefont {Rippin},\ and\ \citenamefont {Knight}}]{vedral1997quantifying}%
  \BibitemOpen
  \bibfield  {author} {\bibinfo {author} {\bibfnamefont {V.}~\bibnamefont {Vedral}}, \bibinfo {author} {\bibfnamefont {M.~B.}\ \bibnamefont {Plenio}}, \bibinfo {author} {\bibfnamefont {M.~A.}\ \bibnamefont {Rippin}},\ and\ \bibinfo {author} {\bibfnamefont {P.~L.}\ \bibnamefont {Knight}},\ }\bibfield  {title} {\bibinfo {title} {Quantifying entanglement},\ }\href {https://doi.org/10.1103/PhysRevLett.78.2275} {\bibfield  {journal} {\bibinfo  {journal} {Physical Review Letters}\ }\textbf {\bibinfo {volume} {78}},\ \bibinfo {pages} {2275} (\bibinfo {year} {1997})}\BibitemShut {NoStop}%
\bibitem [{\citenamefont {Wilde}(2013)}]{wilde2013quantum}%
  \BibitemOpen
  \bibfield  {author} {\bibinfo {author} {\bibfnamefont {M.~M.}\ \bibnamefont {Wilde}},\ }\href {https://doi.org/10.1017/CBO9781139525343} {\emph {\bibinfo {title} {Quantum information theory}}}\ (\bibinfo  {publisher} {Cambridge university press},\ \bibinfo {year} {2013})\BibitemShut {NoStop}%
\bibitem [{\citenamefont {Nielsen}\ and\ \citenamefont {Chuang}(2010)}]{nielsen2010quantum}%
  \BibitemOpen
  \bibfield  {author} {\bibinfo {author} {\bibfnamefont {M.~A.}\ \bibnamefont {Nielsen}}\ and\ \bibinfo {author} {\bibfnamefont {I.~L.}\ \bibnamefont {Chuang}},\ }\href@noop {} {\emph {\bibinfo {title} {Quantum computation and quantum information}}}\ (\bibinfo  {publisher} {Cambridge university press},\ \bibinfo {year} {2010})\BibitemShut {NoStop}%
\bibitem [{\citenamefont {Durrett}(2019)}]{durrett2019probability}%
  \BibitemOpen
  \bibfield  {author} {\bibinfo {author} {\bibfnamefont {R.}~\bibnamefont {Durrett}},\ }\href@noop {} {\emph {\bibinfo {title} {Probability: theory and examples}}},\ Vol.~\bibinfo {volume} {49}\ (\bibinfo  {publisher} {Cambridge university press},\ \bibinfo {year} {2019})\BibitemShut {NoStop}%
\end{thebibliography}%

	\end{document}